\documentclass[10pt, a4paper]{article}%
\usepackage{amsmath,amssymb,eucal}
\usepackage{amsfonts}
\usepackage{mathrsfs}
\usepackage{slashed}
\numberwithin{equation}{section} 
\def\beq{\begin{equation}}
\def\eeq{\end{equation}}

\def\bR{ {{\mathbb{R}}}}
\def\Ric{ {{\rm{Ric}}} }
\def\Tr{ {{\rm{Tr}}} }

\newcommand{\pk}[1]{p_{\kappa}}

\newcommand{\ol}{\overline{l}}
\newcommand{\ul}{\underline{l}}
\newcommand{\oL}{\overline{L}}
\newcommand{\uL}{\underline{L}}
\newcommand{\on}{\overline{n}}
\newcommand{\un}{\underline{n}}

\newcommand{\oV}{\overline{V}}
\newcommand{\uV}{\underline{V}}
\newcommand{\oPsi}{\overline{\Psi}}
\newcommand{\uPsi}{\underline{\Psi}}

\newtheorem{defn}{{\bf Definition}}[section]
\newtheorem{thm}[defn]{{\bf Theorem}}
\newtheorem{cor}[defn]{{\bf Corollary}}
\newtheorem{lem}[defn]{{\bf Lemma}}

\newtheorem{rem}[defn]{{\bf Remark}}

\newtheorem{notation}[defn]{Notation}
\newenvironment{proof}[1][Proof]{\textbf{#1.} }{\hfill \rule{0.5em}{0.5em}}
\begin{document}

\title{Loop representation of Quantum Gravity}
\author{Adrian P. C. Lim \\
Email: ppcube@gmail.com
}

\date{}

\maketitle

\begin{abstract}
A hyperlink is a finite set of non-intersecting simple closed curves in $\mathbb{R}^4 \equiv \bR \times \bR^3$, each curve is either a matter or geometric loop. We consider an equivalence class of such hyperlinks, up to time-like isotopy, preserving time-ordering. Using an equivalence class and after coloring each matter component loop with an irreducible representation of $\mathfrak{su}(2) \times \mathfrak{su}(2)$, we can define its Wilson Loop observable using an Einstein-Hilbert action, which is now thought of as a functional acting on the set containing equivalence classes of hyperlink.

Construct a vector space using these functionals, which we now term as quantum states. To make it into a Hilbert space, we need to define a counting probability measure on the space containing equivalence classes of hyperlinks.

In our previous work, we defined area, volume and curvature operators, corresponding to given geometric objects like surface and a compact solid spatial region. These operators act on the quantum states and by deliberate construction of the Hilbert space, are self-adjoint and possibly unbounded operators.

Using these operators and Einstein's field equations, we can proceed to construct a quantized stress operator and also a Hamiltonian constraint operator for the quantum system. We will also use the area operator to derive the Bekenstein entropy of a black hole.

In the concluding section, we will explain how Loop Quantum Gravity predicts the existence of gravitons, implies causality and locality in quantum gravity, and formulate the principle of equivalence mathematically in its framework.
\end{abstract}

\hspace{.35cm}{\small {\bf MSC} 2010: } 83C45, 81T45, 57R56 \\
\indent \hspace{.35cm}{\small {\bf Keywords}: Loop Quantum gravity, Einstein-Hilbert, Quantum geometry, Causality,    \\
\indent \hspace{2.1cm} Locality, Bekenstein entropy,  Hamiltonian, Stress operator, Gravitons}




\section{Introduction}

After Bernhard Riemann developed non-Euclidean geometry, or more commonly known as Riemannian geometry, he tried to work on a geometric theory of gravity before his death in 1866. He did not succeed, simply because he was only working on 3-dimensional space. It was half a century later that Einstein made use of Riemannian geometry and developed the theory of gravity on 4-dimensional space-time, which is now more commonly known as Generality Relativity (GR).

The turn of the 20th century saw the birth of Quantum Mechanics (QM), which was successful in explaining the orbits of electrons in an atom. The mathematical framework for the formalism of Quantum Mechanics (QM) is based on Functional Analysis, of which its rigorous mathematical foundation is attributed to important researchers like Dirac, Hilbert, Neumann and Weyl.

At around 1927, physicists began to apply QM to fields, resulting in Quantum Field Theories (QFT). This calls for making QFT mathematically consistent, which leads to the emergence of a new branch of mathematics called Constructive Quantum Field Theory (CQFT). The article \cite{2012arXiv1203.3991S} gives a very good account of CQFT. See also \cite{Jaffe}. A sub-branch of CQFT would be Gauge Quantum Field Theory (GQFT), whereby one would quantize classical local gauge theories using path integrals. 

The next natural step would be to quantize gravity. Thiemann in \cite{Thiemann:2002nj, Thiemann:2007zz}, stressed the importance of a background free metric in a quantum theory of gravity, especially in extreme cases such as near a black hole or near the big bang singularity. When the effects of gravity can be ignored, it is reasonable to use the Minkowski metric, which is the metric used in QFT. However, GR resolves to do away with this background metric. Perturbative methods will no longer work when trying to merge GR with QFT together.

The two theories, GR and QM, are incompatible in many ways, and by no means are the reasons we cite below exhaustive. Firstly, time is absolute in QM, but time is relative in GR. Causality is respected in relativity and on this note, GR does not allow instantaneous communication of information. But whether QM respects causality is debatable. A correct quantum theory of gravity should treat time as relative and incorporate causality in its theory. Thus, the notion of absolute time has to be abandoned. We will replace it with time-ordering instead, to be discussed extensively throughout this article.

Secondly, GR gives us a deterministic geometric description of space-time, whereas QM is a non-deterministic theory. We cannot keep both characteristics, and it is obvious that a quantum theory of gravity must be a non-deterministic description of space-time. Thus, we must have a topological theory to form the mathematical foundation of a quantum theory of gravity. But bear in mind that if we have a topological theory, then there is no distinction between space and time, hence causality will not be respected. Our proposed idea would be to consider an equivalence class of geometrical objects in the theory, in place of topological objects. See Section \ref{s.obs}.

Finally, GR is a local theory, while QM is a global theory. It would seem that it is not possible to reconcile both together and we have to choose one in favor of the other. But, we will see that it is possible to keep both aspects in the theory. See subsection \ref{ss.lpe}.

In \cite{Thiemann:2002nj}, Thiemann explained that there are major inconsistencies in QFT and GR. In QFT, when distances are short, GR predicted the conversion of virtual particles into black holes. The quantization of volume will forbid such phenomena from happening. See Section \ref{ss.vo}. Einstein's equations will lead to singularities in the metric, for example the metric in the Friedmann-Robertson-Walker (FRW) universe model has a singular metric when we approach beginning of time. Quantum gravity forbids this from happening. See Section \ref{ss.co}.

Finally there is another inconsistency, which is in Einstein's equations, given by Equation (\ref{e.einstein}). There are problems with this equation, as explained again in \cite{Thiemann:2002nj}. The metric enters into the computation of the stress-energy tensor, which in QFT, is computed using the Minkowski metric. However, Einstein's equations are used to solve for the metric, and when the space is strongly curved, the correct metric is no longer the Minkowski metric. Hence in such a regime, QFT is no longer valid. Furthermore, the LHS of the equation is geometry or classical theory, whereas the RHS is matter interaction or quantum theory. Mathematically, the field equations are also incompatible, in the sense that the LHS is interpreted as a tensor, whereas the RHS should be interpreted as an operator.  Therefore, he called for writing the equation as an operator equation and solve for a quantized metric. However, we cannot concur with him on this. The metric is a dynamical variable, not an observable to be quantized. To rectify this inconsistency, we will instead use Einstein's equations to quantize the stress-energy tensor. See Section \ref{s.hmo}.

QM and GR have indeed revolutionized physics as we know it and both have changed the way we view the world. See \cite{rovelli2004quantum}. To understand Planck scale physics (See \cite{Rovelli1998}.), and with the reasons cited above, it is indeed necessary to have a quantum theory of gravity.

But what would a theory of quantum gravity look like? A quantum theory of gravity should be independent of background metric (See \cite{Smolin2006-SMOTCF}.), and also it should be invariant under diffeomorphism of the ambient space-time manifold, which we will choose to be $\bR \times \bR^3$, in this article. These are the main ingredients quantum gravity should have, as explained by Thiemann.

A potential candidate for quantum gravity is M-theory, which is a topological theory. This theory lives in ten dimensions and is necessarily supersymmetric. Unfortunately, it is not background metric independent, as pointed out in \cite{Ashtekar:2000eq}. Another issue with M-theory is that it is difficult to maintain locality and causality together, as explained in \cite{THOOFT2001157}. Thiemann in \cite{Thiemann:2007zz} also gave reasons to support his own reservations.

Witten argued in \cite{MR990772} that Quantized Chern-Simons Theory (QCS) is related to knot theory. This may arguably be the first paper that gives rise to Topological Quantum Field Theory (TQFT). Now, QCS is actually a GQFT, using the classical Chern-Simons action, which is independent of any metric used, to construct Wilson Loop observables. Indeed it does give us knot invariants that one is familiar with. See \cite{CS-Lim02}. We can use a path integral, but using the Yang-Mills action, to quantize the Yang-Mills Theory into Quantized Yang-Mills Theory (QYM). Unfortunately, it is not a TQFT, as its construction is metric dependent.

The Einstein-Hilbert action in GR will yield Einstein's field equations through the principle of least action. See \cite{Witten:1988hc}. A quantized theory of Einstein-Hilbert Theory, should also be a TQFT, as a successful quantization of GR should be metric independent. This insight has already been discussed in \cite{PMIHES_1988__68__175_0} and \cite{atiyah_1990}. But is there such a quantized Einstein-Hilbert Theory? The answer is yes, and is given by Loop Quantum Gravity (LQG), which has been around for decades.

A path integral of the form Expression \ref{e.eh.1} given later, is independent of background metric. And we will use this expression to construct a loop representation for quantum gravity, first mooted in \cite{PhysRevLett.61.1155}. In this LQG model, one can actually see how matter and space-time interact, on equal pedestal, details to be given later. This is referred to as quantum geometry in \cite{Ashtekar:2000eq}. Matter is represented by matter hyperlinks, to be defined in Section \ref{s.cvs}, whereas space-time is represented by geometric hyperlinks, surfaces and compact solid regions. This is also highlighted in \cite{Thiemann:2002nj}. Furthermore, the theory implies causality and predicts the existence of gravitons, which has to play an important role in quantum gravity. See Section \ref{s.fr}.

\section{Spin Networks}\label{s.sn}

LQG has existed for decades and many authors have written good introductory works, such as \cite{rovelli2004quantum, Rovelli1998, Mercuri:2010xz,  Ashtekar:2004vs}, on this subject and giving reasons for the necessity for a quantum theory of gravity.
A mathematical formulation of quantum gravity in a 3-manifold can be found in \cite{Thiemann:2007zz, 1990NuPhB.331...80R}. We would also like to mention an expository article \cite{Carlip:1995zj} on $2+1$ dimensional quantum gravity, which is related to Chern-Simons theory, as explained in \cite{Witten:1988hc}.

By no means is this list exhaustive and we apologize if we missed out any other suitable references. From this list of references, people will generally associate LQG with spin networks in a 3-manifold and canonical quantization of GR. It is also synonymous with  Quantum Spin Dynamics (QSD) or Canonical Quantum Gravity (CQG) (See \cite{rovelli2004quantum}.), which we think are more appropriate designations.

The idea of a spin network was first conceived by Penrose. A spin network is essentially a graph, each vertex has valency 3 and a spin is assigned to each edge, satisfying a certain inequality for edges incident on a common vertex. We refer the reader to \cite{Baez:1999sr, Baez:1999:Online, Baez1996253, Rovelli:1995ac} for a more detailed description.

The role of a spin network in quantum gravity is that it discretizes space-time and approximates the space of connections by treating it as a finite product copies of a (compact) gauge group. This allows us to define a Haar measure on this finite dimensional space. By taking the union of all possible spin networks, one defines a Hilbert space of functionals on connection. This construction was made rigorous in \cite{Baez1996253}. On a separate note, Ashtekar and Lewandowski defined a generalized measure on the space of connections. Such a measure is also defined using cylindrical measures or cylinder functions. See \cite{Thiemann:2002nj, Thiemann:2007zz}.

Now, the space of connections modulo gauge transformations is an infinite dimensional space, so there is no notion of Lebesgue measure or even Haar measure that can possibly be defined on it. Thus, the construction as described earlier do not define a measure. As such, it does not describe the loop representation as described in \cite{PhysRevLett.61.1155}. Furthermore, the spin networks do not define a suitable set of functionals on the space of connections. Hence, it does not describe the conjugate self-dual representation, also described in \cite{PhysRevLett.61.1155}.

Another argument against using the spin networks to quantize gravity is that it does not take into account of the Lorentz group, or the Poincare group. The spin numbers assigned to the edges only represent the angular momentum. In addition, spin networks are good for describing gravity in 3-manifold; to describe gravity in 4-manifold, one has to consider spin foam. See \cite{Baez:1999sr}.

The Hamiltonian formulation of GR will lead to define physically admissible states, defined by three constraints, namely the Gauss constraint, spatial diffeomorphism constraint and the Hamiltonian constraint. See \cite{Ashtekar:2000eq, doi:10.1063/1.1587095, Thiemann:2006cf}. The Gauss constraint is associated with ${\rm SU}(2)$ gauge transformations. The diffeomorphism constraint is associated with the invariance of GR under spatial diffeomorphisms. The Hamiltonian constraint is associated with the invariance of GR under diffeomorphisms of the spatial surface (region) for 1+2 (1+3) dimensional quantum gravity.

Because our ambient space is $\bR^4 \equiv \bR \times \bR^3$, we see that we can foliate $\bR^4$ into a set containing hypersurfaces $\{\Sigma_t:\ t \in \bR\}$. See \cite{Thiemann:2007zz}. Explicitly, we have a diffeomorphism $\phi: \bR \times \bR^3 \rightarrow \bR \times \bR^3$, whereby $(x_0, x) \mapsto \phi(x_0, x)$, with $\Sigma_t := \{\phi(t, x):\ x \in \bR^3\}$. This diffeomorphism $\phi$ is completely arbitrary and thus, the set of foliations is in one to one correspondence with the set of diffeomorphisms of $\bR \times \bR^3$. Let $\{S_a(t,p):\ a = 1,2,3\}$ be a set of tangent vector fields at $p \in \Sigma_t$ and by introducing a metric $g$, let $n$ be a unit normal vector to $\Sigma_t$ at $p$. The vector $\partial_{x_0} \equiv \partial/\partial x_0$ is being push forward by $\phi$ to be $\phi_\ast \partial_{x_0}$ and we can write \beq \phi_\ast \partial_{x_0} = Nn + U^a S_a,\ N \neq 0. \nonumber \eeq Now, $N$ and $U^a S_a$ are called the lapse function and shifted vector fields respectively. The lapse function $N$ is associated with the Hamiltonian constraint; the shifted vector fields will be associated with the spatial diffeomorphism constraint, as explained in \cite{Thiemann:2002nj}.

Together, the spatial diffeomorphism constraint and the Hamiltonian constraint are consequences of the invariance of GR under four-dimensional diffeomorphism. But there is a reason why we consider the Hamiltonian constraint separately. The diffeomorphism in the normal direction $n$ or orthogonal to the hypersurface means evolution in the time parameter.

Let us return back to the lapse function $N$ and since it cannot vanish, we assume that it is always positive. By foliating the space, we see that there is a well-defined notion of a flow of time, thus breaking the symmetry between space and time. This is crucial in the quantization of gravity, because later it will allow us to apply axial gauge fixing. As a consequence, it also allows us to order the hypersurfaces, i.e. $\Sigma_t < \Sigma_\tau$ if $t < \tau$.

The usual metric canonical coordinates are used in the ADM formulation of GR. Later, Ashtekar introduced (density-valued) soldering and $\mathfrak{su}(2)$-valued one forms in 3-space, to replace the metric variables. See \cite{PhysRevLett.57.2244, PhysRevD.36.1587}. These replacement variables are henceforth referred to as Ashtekar variables, and the Gauss and Hamiltonian constraints can be rewritten in these new variables. One advantage in this approach, is that Thiemann used a generalization of Ashtekar's formalism and promote the Hamiltonian constraint to a quantum operator in the LQG framework. See \cite{THIEMANN1996257}. We will also define the Hamiltonian constraint operator later in Section \ref{s.hmo}.

In any quantum theory, operators corresponding to physical observables has to be constructed. In \cite{rovelli1995discreteness}, the authors used canonical quantization on Ashtekar variables to quantize area and volume, for 3-dimensional quantum gravity. For a good expository account of canonical quantization, the reader can refer to \cite{doi:10.1063/1.1587095}. This procedure will work fine, but it fails when applied to curvature. The reason is because curvature involves derivatives of the Ashtekar variables and one will run into problems trying to apply canonical quantization.

In the quantization of area and volume, the authors in the said article used spin networks as the states for which these operators act on. The confusion comes in as spin networks will then appear to play multiple roles in connection and loop representations of quantum gravity. See the section on mathematical foundations of quantum gravity in \cite{Thiemann:2007zz}. To resolve this issue, the authors in \cite{Rovelli:1995ac} described a mapping from the connection representation to loop representation.

When implementing canonical quantization, one has to ensure that the three constraints as discussed above are all satisfied, which will lead to some techinical difficulties. See \cite{Thiemann:2006cf}.  Using the spin network description in QSD or CQG, will not incorporate diffeomorphism invariance inside the theory. To obtain diffeomorphism invariance, a diffeomorphism constraint has to be imposed on this space of functionals. But there are some mathematical difficulties trying to implement this procedure. See \cite{Thiemann:2002nj, Ashtekar:2000eq, 0264-9381-21-15-R01}.

\begin{rem}
Such a problem does not exist if one uses path integral quantization.
\end{rem}

One possible remedy would be to construct a set of loops from a spin network, as was done in \cite{Rovelli:1995ac}. The reason for considering loops is due to \cite{PhysRevLett.61.1155}, whereby the authors described the connection between quantum gravity and link invariants. Any quantized theory which yield topological invariants will be invariant under diffeomorphism of the ambient space.

With all the above reasons as discussed above, it would be better to be consistent, and work on the loop representation of quantum gravity in $\bR \times \bR^3$. The objects we work with are hyperlinks, which is a set of non-intersecting, simple closed curves in $\bR \times \bR^3$. In \cite{CS-Lim02}, we showed how a link in $\bR^3$ can be projected to form a planar graph, with each vertex having valency 4, and hence compute the linking number from it. In \cite{EH-Lim03}, we defined the hyperlinking number of a hyperlink, and showed how one can compute Wilson Loop observables from it. In \cite{EH-Lim06}, we showed that the hyperlinking number is equal to the linking number of the projection of a hyperlink to form a planar graph, up to $\pm 1$, provided we define a time-ordering (Definition \ref{d.tau.1}) for the component loops. Under time-like isotopy (Definition \ref{d.ts.1}) and time ordering (Definition \ref{d.tl.2}), the hyperlinking number of a hyperlink will be an invariant.

When one considers loop representation of quantum gravity, all three constraints are satisfied, as explained in \cite{doi:10.1063/1.1587095}. In \cite{0264-9381-21-15-R01}, the authors mentioned using an averaging process to impose the diffeomorphism constraint on the quantum Hilbert space. For the case involving the Hamiltonian constraint, the reader can also refer to \cite{JACOBSON1988295}.

\begin{rem}
We will explain later how we implement the Gauss constraint in Item \ref{r.tr.4a} in Remark \ref{r.tr.4}, the diffeomorphism constraint in Item \ref{r.tr.3a} and the Hamiltonian constraint in Item \ref{r.tr.3b}, both in Remark \ref{r.tr.3}.
\end{rem}

A path integral  is essentially, an averaging process over all possible dynamical variables considered in the classical theory. Path integrals are used extensively in CQFT. See \cite{2012arXiv1203.3991S}. We would like to mention the work by the authors in \cite{PhysRevLett.61.1155}, who were the first to write down a path integral for quantum gravity. As pointed out in \cite{Thiemann:2007zz}, path integrals has the advantage that they are diffeomorphism invariant, but unfortunately they are ill-defined.

The path integrals involving area, volume and curvature, were all defined and computed in \cite{EH-Lim03}, \cite{EH-Lim04} and \cite{EH-Lim05} respectively. Each of these path integrals can be explicitly computed using topological invariants defined in \cite{EH-Lim06}, hence consistent with the view point that geometric notions should play a central role in LQG, as stated in \cite{Ashtekar:2000eq}. Thus, we now have a consistent way of quantizing physical observables into operators, which we will summarize the results in Section \ref{s.qo}.

Diffeomorphism invariance in a theory is immediate once topological invariants are obtained. Compare this with canonical quantization in QSD or CQG, whereby one has to enforce a diffeomorphism constraint. In this article, we will see that area, volume and curvature will be quantized into operators and their respective eigenvalues have a physical interpretation, to be explained in the later sections, something which canonical quantization is not able to achieve.

\section{Spin representation}\label{s.su.1}

\begin{rem}
Throughout this article, we adopt Einstein's summation convention, i.e. we sum over repeated superscripts and subscripts. Indices such as $a,b,c, d$ and greek indices such as $\mu,\gamma, \alpha, \beta$ will take values from 0 to 3; indices labeled $i, j, k$, $\bar{i}, \bar{j}, \bar{k}$ will only take values from 1 to 3.
\end{rem}

Let $\mathfrak{su}(2)$ be the Lie Algebra of ${\rm SU}(2)$. Here, we summarize the representation of the following Lie Algebra $\mathfrak{su}(2) \times \mathfrak{su}(2)$ as described in \cite{EH-Lim02}. This direct product of Lie Algebras inherit the Lie bracket from $\mathfrak{su}(2)$ in the obvious way.

Let $\{\breve{e}_1, \breve{e}_2, \breve{e}_3\}$ be any basis for the first copy of $\mathfrak{su}(2)$ and
$\{\hat{e}_1, \hat{e}_2, \hat{e}_3\}$ be any basis for the second copy of $\mathfrak{su}(2)$, satisfying the conditions
\begin{align*}
[\breve{e}_1, \breve{e}_2] =& \breve{e}_3,\ \ [\breve{e}_2, \breve{e}_3] = \breve{e}_1,\ \ [\breve{e}_3, \breve{e}_1] = \breve{e}_2, \\
[\hat{e}_1, \hat{e}_2] =& \hat{e}_3,\ \ [\hat{e}_2, \hat{e}_3] = \hat{e}_1,\ \ [\hat{e}_3, \hat{e}_1] = \hat{e}_2.
\end{align*}

Using this basis, define \beq \mathcal{E}^+ = \sum_{i=1}^3\breve{e}_{i}\ ,\ \mathcal{E}^- = \sum_{i=1}^3\hat{e}_{i}. \nonumber \eeq

Let
\beq \hat{E}^{01} = (\breve{e}_1,0),\ \hat{E}^{02} = (\breve{e}_2,0) ,\ \hat{E}^{03} = (\breve{e}_3,0)  \nonumber \eeq and \beq \hat{E}^{23} = (0,\hat{e}_1) ,\ \hat{E}^{31} = (0,\hat{e}_2) ,\ \hat{E}^{12} = (0,\hat{e}_3). \nonumber \eeq Do note that $\hat{E}^{\alpha\beta} = -\hat{E}^{\beta\alpha}$ and write
\beq \hat{E}^{\tau(1)} = \hat{E}^{23}, \ \ \hat{E}^{\tau(2)} = \hat{E}^{31}, \ \ \hat{E}^{\tau(3)} = \hat{E}^{12}, \nonumber \eeq all in $\mathfrak{su}(2) \times \mathfrak{su}(2)$. Finally, denote \beq \mathcal{E} :=
\left(
  \begin{array}{cc}
    \mathcal{E}^+ &\ 0 \\
    0 &\ -\mathcal{E}^- \\
  \end{array}
\right) \cong (\mathcal{E}^+, \mathcal{E}^-) \in \mathfrak{su}(2) \times \mathfrak{su}(2). \nonumber \eeq

Let $\rho^\pm: \mathfrak{su}(2) \rightarrow {\rm End}(V^\pm)$ be an irreducible finite dimensional representation, indexed by half-integer and integer values $j_{\rho^\pm} \geq 0$. The representation $\rho: \mathfrak{su}(2) \times \mathfrak{su}(2) \rightarrow {\rm End}(V^+) \times {\rm End}(V^-)$ will be given by $\rho = (\rho^+, \rho^-)$, with \beq \rho: \alpha_i\hat{E}^{0i} + \beta_j \hat{E}^{\tau(j)} \mapsto \left(\sum_{i=1}^3\alpha_i \rho^+(\breve{e}_i), \sum_{j=1}^3\beta_j \rho^-(\hat{e}_j) \right). \nonumber \eeq By abuse of notation, we will now write $\rho^+ \equiv (\rho^+, 0)$ and $\rho^- \equiv (0, \rho^-)$ in future and thus $\rho^+(\hat{E}^{0i}) \equiv \rho^+(\breve{e}_i)$, $\rho^-(\hat{E}^{\tau(j)}) \equiv \rho^-(\hat{e}_j)$.

Without loss of generality, we assume that $\rho^\pm(\hat{E})$ is skew-Hermitian for any $\hat{E} \in \mathfrak{su}(2)$ and that $V^\pm$ is a finite dimensional inner product space. Note that the dimension of $V^\pm$ is given by $2j_{\rho^\pm} + 1$. Then it is known that the Casimir operator is
\begin{align*}
\sum_{i=1}^3 \rho^+(\hat{E}^{0i})\rho^+(\hat{E}^{0i}) \equiv \sum_{i=1}^3 \rho^+(\breve{e}_i) \rho^+(\breve{e}_i) = -\xi_{\rho^+} I_{\rho^+},\\
\sum_{i=1}^3 \rho^-(\hat{E}^{\tau(i)})\rho^-(\hat{E}^{\tau(i)}) \equiv \sum_{i=1}^3 \rho^-(\hat{e}_i) \rho^-(\hat{e}_i)= -\xi_{\rho^-} I_{\rho^-},
\end{align*}
$I_{\rho^\pm}$ is the identity operator for $V^\pm$ and $\xi_{\rho^\pm} := j_{\rho^\pm}(j_{\rho^\pm}+1)$. In QM, $\sqrt{\xi_{\rho^\pm}}$ is interpreted as the total momentum of a particle, given the representation $\rho^\pm$.

Without any loss of generality, by choosing a suitable basis in $V^+$, we will always assume that $\rho^+(i\mathcal{E}^+)$ is diagonal, with the real eigenvalues given by the set
\beq \left\{
\begin{array}{ll}
\{\pm \breve{\lambda}_1, \pm \breve{\lambda}_2, \cdots, \pm \breve{\lambda}_{(2j_{\rho^+}+1)/2}\}, & \hbox{$2j_{\rho^+} + 1$ is even;} \\
\{\pm \breve{\lambda}_1, \pm \breve{\lambda}_2, \cdots, \pm \breve{\lambda}_{j_{\rho^+}}, 0\}, & \hbox{$2j_{\rho^+} + 1$ is odd.}
\end{array}
\right. \nonumber \eeq If we assume that $\breve{\lambda}_k > 0$ and $\breve{\lambda}_1< \breve{\lambda}_2 < \cdots$ is increasing, then we have $\breve{\lambda}_{k+1} - \breve{\lambda}_k = \sqrt3$ for each $k$. That is, the eigenvalues are equally spaced apart by a constant $\sqrt3$. In QM, these eigenvalues are related to the $z$-component of the momentum of a particle, given the representation $\rho^+$.

Similarly, by choosing another suitable basis in $V^-$, we will always assume that $\rho^-(i\mathcal{E}^-)$ is diagonal, with the real eigenvalues given by the set
\beq \left\{
\begin{array}{ll}
\{\pm \hat{\lambda}_1, \pm \hat{\lambda}_2, \cdots, \pm \hat{\lambda}_{(2j_{\rho^-}+1)/2}\}, & \hbox{$2j_{\rho^-} + 1$ is even;} \\
\{\pm \hat{\lambda}_1, \pm \hat{\lambda}_2, \cdots, \pm \hat{\lambda}_{j_{\rho^-}}, 0\}, & \hbox{$2j_{\rho^-} + 1$ is odd.}
\end{array}
\right. \nonumber \eeq  We can interpret these eigenvalues in a similar manner.

We have
\begin{align}
\label{e.tr.1}
\begin{split}
\Tr\ \rho(e^{i\mathcal{E}^{+}}) =&
\left\{
  \begin{array}{ll}
    \sum_{v=1}^{(2j_{\rho^+}+1)/2}2\cosh(\breve{\lambda}_v), & \hbox{$2j_{\rho^+} + 1$ is even;} \\
    1 + \sum_{v=1}^{j_{\rho^+}}2\cosh(\breve{\lambda}_v), & \hbox{$2j_{\rho^+} + 1$ is odd;}
  \end{array}
\right. \\
\Tr\ \rho(e^{i\mathcal{E}^{-}}) =&
\left\{
\begin{array}{ll}
\sum_{v=1}^{(2j_{\rho^-}+1)/2}2\cosh(\hat{\lambda}_v), & \hbox{$2j_{\rho^-} + 1$ is even;} \\
1 + \sum_{v=1}^{j_{\rho^-}}2\cosh(\hat{\lambda}_v), & \hbox{$2j_{\rho^-} + 1$ is odd.}
\end{array}
\right.
\end{split}
\end{align}

In either case, we have $\Tr\ \rho^{\pm}(e^{i\mathcal{E}^{\pm}}) \geq 1$ and hence $\log\ \Tr\ \rho^{\pm}(e^{i\mathcal{E}^{\pm}}) \geq 0$, for any irreducible representation. Finally, note that $\Tr\ \rho^{\pm}(e^{i\mathcal{E}^{\pm}})$ is well-defined, even though $\mathcal{E}^\pm$ is not in general.


\section{Time-like hyperlink}\label{s.tls}


As stated in \cite{Thiemann:2002nj}, according to a theorem by Geroch, any globally hyperbolic 4-manifold has to be of the form $\bR \times M$, whereby $M$ is a 3-manifold. See \cite{Geroch:1970uw}. The simplest but no doubt important 3-manifold to be considered will be $\bR^3$. Thus in this article, the 4-manifold in consideration will be $\bR^4 \cong \bR \times \bR^3$, whereby $\bR$ will be referred to as the time axis and $\bR^3$ is the spatial 3-dimensional Euclidean space. Fix the standard coordinates on $\bR^4$, $(x_0, x_1, x_2, x_3)$, whereby $x_0$ will be referred to as time.

In future, when we write $\bR^3$, we refer to the spatial subspace in $\bR^4$. Let $\pi_0: \bR^4 \rightarrow \bR^3$ denote this projection. Let $\{e_i\}_{i=1}^3$ be the standard basis in $\bR^3$. And $\Sigma_i$ is the plane in $\bR^3$, containing the origin, whose normal is given by $e_i$. So, $\Sigma_1$ is the $x_2-x_3$ plane, $\Sigma_2$ is the $x_3-x_1$ plane and finally $\Sigma_3$ is the $x_1-x_2$ plane. Note that $\bR \times \Sigma_i \cong \bR^3$ is a 3-dimensional subspace in $\bR^3$ and let $\pi_i: \bR^4 \rightarrow \bR \times \Sigma_i$ denote this projection.

For a finite set of non-intersecting simple closed curves in $\bR^3$ or in $\bR \times \Sigma_i$, we will refer to it as a link. If it has only one component, then this link will be referred to as a knot. A simple closed curve in $\bR^4$ will be referred to as a loop. A finite set of non-intersecting loops in $\bR^4$ will be referred to as a hyperlink in this article. If we assign an orientation to each component loop, then the hyperlink is said to be oriented.

We need to consider the space of hyperlinks in $\bR \times \bR^3$, which is too big for our consideration. Given any hyperlink in $\bR \times \bR^3$, it is ambient isotopic to the trivial hyperlink, a disjoint union of `unlinked' loops. Hence, the equivalence class of ambient isotopic hyperlinks will give us only the trivial hyperlink. Therefore, we will instead consider a special equivalence class of hyperlinks.

\begin{defn}(Time-like hyperlink)\label{d.tl.1}\\
Let $L$ be a hyperlink. We say it is a time-like hyperlink if given any 2 distinct points $\vec{p}\equiv (x_0, x_1, x_2, x_3), \vec{q}\equiv (y_0, y_1, y_2, y_3) \in L$, $\vec{p} \neq \vec{q}$,
\begin{enumerate}
  \item (T1) $\sum_{i=1}^3(x_i - y_i)^2 > 0$;
  \item (T2) if there exist $i, j$, $i \neq j$ such that $x_i = y_i$ and $x_j = y_j$, then $x_0 - y_0 \neq 0$.
\end{enumerate}
\end{defn}

We make the following remarks, which is immediate from the definition.

\begin{rem}\label{r.p.1}
\begin{enumerate}
  \item We refer the reader to \cite{EH-Lim06} for the reason why we use the term `time-like'.
  \item In Condition T1, we insist that any 2 distinct points in a hyperlink are also separated, when projected using $\pi_0$, in 3-dimensional spatial space $\bR^3$. This is to ensure that when we project the hyperlink in $\bR^3$, we obtain a link.
  \item Conditions T1 and T2 imply that given a hyperlink $L$, for each $i=1,2,3$, $\pi_i(L) \in \bR \times \Sigma_i$ is a link.
\end{enumerate}
\end{rem}

\begin{defn}\label{d.ts.1}
Two oriented time-like hyperlinks $L$ and $L'$ in $\bR \times \bR^3$ are time-like isotopic to each other if there is an orientation preserving continuous map $F: \bR \times \bR^3 \times [0,1] \rightarrow \bR \times \bR^3$, such that
\begin{enumerate}
  \item $F_0$ is the identity map;
  \item $F_t$ is a homeomorphism from $\bR \times \bR^3$ to $\bR \times \bR^3$;
  \item $F_1(L) = L'$;
  \item each $F_t(L)$ is a time-like hyperlink.
\end{enumerate}
In other words, two time-like hyperlinks $L_1$ and $L_2$ in $\bR \times \bR^3$ are time-like isotopic if $L_1$ can be continuously deformed to $L_2$ while remaining time-like.
 \end{defn}

\begin{rem}
By definition, $L$ and $L'$ in $\bR \times \bR^3$ are time-like isotopic to each other implies that
\begin{enumerate}
  \item $\pi_0(L)$ and $\pi_0(L')$ are ambient isotopic to each other in $\bR^3$;
  \item $\pi_i(L)$ and $\pi_i(L')$ are ambient isotopic to each other in $\bR \times \Sigma_i$, $i=1,2,3$.
\end{enumerate}
\end{rem}

Throughout this article, all our hyperlinks will be time-like and we consider equivalence classes of such hyperlinks using Definition \ref{d.ts.1}.

\begin{defn}\label{d.tau.1}
For any $\vec{p} = (p_0, p_1, p_2, p_3) \in \bR \times \bR^3$, we define $\tau(\vec{p}) = p_0$.
Given two loops $\ol$ and $\ul$, we say that $\ol < \ul$ if for any $\vec{p} \in \ol$, any $\vec{q} \in \ul$, we have that $\tau(\vec{p}) < \tau(\vec{q})$. When $\ol < \ul$, we say that the loop $\ol$ occurs before the loop $\ul$. If $\ol > \ul$, we say that the loop $\ol$ occurs after the loop $\ul$.
\end{defn}

\section{Construction of a Vector space of functionals}\label{s.cvs}

Let $V \rightarrow \bR \times \bR^3$ be a 4-dimensional vector bundle, with structure group $SO(3,1)$. This implies that $V$ is endowed with a metric, $\eta^{ab}$, of signature $(-, +, +, +)$, and a volume form $\epsilon_{abcd}$. Here, $\epsilon^{\mu \gamma \alpha \beta} \equiv \epsilon_{\mu \gamma \alpha \beta}$ is equal to 1 if the number of transpositions required to permute $(0123)$ to $(\mu\gamma\alpha\beta)$ is even; otherwise it takes the value -1.

Suppose $V$ has the same topological type as $T\bR^4$, so that isomorphisms between $V$ and $T\bR^4$ exist. Hence we may assume that $V$ is a trivial bundle over $\bR^4$. Without loss of generality, we will assume the Minkowski metric $\eta^{ab}$ is given by \beq \eta = -dx_0 \otimes dx_0 + \sum_{i=1}^3 dx_i \otimes dx_i. \nonumber \eeq

However, there is no natural choice of an isomorphism. A vierbein $e$ is a choice of isomorphism between $T\bR^4$ and $V$. It may be regarded as a $V$-valued one form, obeying a certain condition of invertibility. A spin connection $\omega_{\mu \gamma}^a$ on $V$, is anti-symmetric in its indices $\mu$, $\gamma$. It takes values in $\Lambda^2 (V)$, whereby $\Lambda^k (V)$ denotes the $k$-th antisymmetric tensor power or exterior power of $V$. The isomorphism $e$ and the connection $w$ can be regarded as the dynamical variables of General Relativity.

The curvature tensor is defined as
\beq R_{\mu\gamma}^{ab} = \partial_a \omega_{\mu \gamma}^{b} - \partial_b \omega_{\mu \gamma}^{a} + [\omega^a, \omega^b]_{\mu\gamma},\ \partial_a \equiv \partial/\partial x_a, \nonumber \eeq
or as $R = d\omega + \omega \wedge \omega$. It can be regarded as a two form with values in $\Lambda^2 (V)$.

Using the above notations, the Einstein-Hilbert action is written as \beq S_{{\rm EH}}(e, \omega) =\frac{1}{8}\int_{\bR^4} e \wedge e \wedge R =\frac{1}{8}\int_{\bR^4}\epsilon^{\mu \gamma \alpha \beta}\epsilon_{abcd}\ e_\mu^{a}e_\gamma^b R_{\alpha\beta}^{cd}. \label{e.eh.2} \eeq The expression $e \wedge e \wedge R$ is a four form on $\bR^4$ taking values in $V \otimes V \otimes \Lambda^2(V)$, which maps to $\Lambda^4(V)$. But $V$ with the structure group $SO(3,1)$ has a natural volume form, so a section of $\Lambda^4(V)$ may be canonically regarded as a function. Thus Equation (\ref{e.eh.2}) is an invariantly defined integral.
By varying Equation (\ref{e.eh.2}) with respect to $e$, we will obtain the Einstein equations in vacuum. See \cite{Witten:1988hc}.

The metric $\eta^{ab}$ on $V$, together with the isomorphism $e$ between $T\bR^4$ and $V$, gives a (non-degenerate) metric $g^{ab} = e_\mu^a e_\gamma^b \eta^{\mu\gamma}$ on $T\bR^4$. By varying Equation (\ref{e.eh.2}) with respect to the connection $\omega$, we will obtain an equation that identifies $\omega$ as the Levi-Civita connection associated with the metric $g^{ab}$.

Let $\overline{\mathcal{S}}_\kappa(\bR^4) \subset L^2(\bR^4)$ be a Schwartz space, which consists of functions of the form $f = p \cdot \sqrt{\phi_\kappa}$, whereby $p$ is a polynomial on $\bR^4$, and $\phi_\kappa(\vec{x}) = \frac{\kappa^4}{(2\pi)^2}e^{-\kappa^2|\vec{x}|^2/2}$ is the Gaussian function with variance $1/\kappa^2$.

Using the standard coordinates on $\bR^4$, let $\Lambda^1(\bR^3)$ denote the subspace in $\Lambda^1(\bR^4)$ spanned by $\{dx_1, dx_2, dx_3\}$. Define
\begin{align}
\begin{split}
L_\omega :=& \overline{\mathcal{S}}_\kappa(\bR^4) \otimes \Lambda^1(\bR^3)\otimes \mathfrak{su}(2) \times \mathfrak{su}(2), \\
L_e :=& \overline{\mathcal{S}}_\kappa(\bR^4) \otimes \Lambda^1(\bR^3)\otimes V.
\end{split}\label{e.o.1}
\end{align}

Any spin connection $\omega$ as above can be written as $\omega \equiv A^a_{\alpha\beta} \otimes dx_a\otimes \hat{E}^{\alpha\beta}$, whereby $A^a_{\alpha\beta}: \bR^4 \rightarrow \bR$ is smooth and we identify $\Lambda^2(V)$ with $\mathfrak{su}(2) \times \mathfrak{su}(2)$. Considering that this space of smooth spin connections is too big for our purpose, we need to `trim' down this space. Under axial gauge fixing, every spin connection can be gauge transformed into \beq \omega = A^i_{\alpha\beta} \otimes dx_i \otimes \hat{E}^{\alpha\beta}\in \overline{\mathcal{S}}_\kappa(\bR^4) \otimes \Lambda^1(\bR^3)\otimes \mathfrak{su}(2) \times \mathfrak{su}(2) \equiv L_\omega,\nonumber \eeq $A^i_{\alpha\beta}: \bR^4 \rightarrow \bR$ smooth, subject to the conditions \beq A^1_{\alpha\beta}(0, x^1, 0, 0) = 0,\ A^2_{\alpha\beta}(0, x^1, x^2, 0)=0,\ A^3_{\alpha\beta}(0, x^1, x^2, x^3) = 0. \label{e.r.1} \eeq There is an implied sum over repeated indices. We will drop all these restrictions on $A^i_{\alpha\beta}$.

\begin{rem}\label{r.lqg.1}
\begin{enumerate}
  \item Note that $A^{i}_{\alpha\beta} = -A^{i}_{\beta\alpha}\in \overline{\mathcal{S}}_\kappa(\bR^4)$.
  \item A $\mathfrak{su}(2)$-valued spin connection was also considered in \cite{PhysRevD.73.124038}.
\end{enumerate}
\end{rem}

Recall $e$ is $V$-valued one form. Let $\{E^\gamma\}_{\gamma=0}^3$ be a basis for $V$, using the Minkowski metric for the 4-dimensional vector space $V$. Even though $e$ is not a gauge, we will still apply axial gauge fixing argument as above. After applying gauge fixing on $e$, we consider $e: T\bR^4 \rightarrow V$, written as \beq e = B^i_\gamma \otimes dx_i \otimes E^\gamma \in \overline{\mathcal{S}}_\kappa(\bR^4) \otimes \Lambda^1(\bR^3) \otimes V  \equiv L_e. \nonumber \eeq There is an implied sum over repeated indices. Observe that $e$ is now no longer invertible. We may interpret $E^\alpha$ as generator for translation in the $E^\alpha$ direction, which corresponds to translation in the Poincare group. Note that the Lorentz group is a Lie subgroup of the Poincare Lie group. If we regard $\{\hat{E}^{0i}: i=1,2,3\}$ and $\{\hat{E}^{\tau(i)}: i=1,2,3\}$ as generators of boost and rotation respectively, then we may interpret $\{\omega, e\}$ as a connection with values in the Poincare Lie Algebra.

After applying axial gauge fixing to $\omega$ and $e$, the Einstein-Hilbert action is then computed as ($\partial_0 \equiv \partial/\partial x_0$)
\begin{align*}
S_{{\rm EH}}(e, \omega) =
\frac{1}{8}&\int_{\bR^4}\epsilon^{abcd}B^1_\gamma B^2_\mu[\hat{E}^{\gamma \mu}]_{ab} \cdot \partial_0 A^3_{\alpha\beta}[\hat{E}^{\alpha\beta}]_{cd} dx_1\wedge dx_2 \wedge dx_0 \wedge dx_3\\
+& \frac{1}{8}\int_{\bR^4}\epsilon^{abcd}B^2_\gamma B^3_\mu[\hat{E}^{\gamma \mu}]_{ab} \cdot \partial_0 A^1_{\alpha\beta}[\hat{E}^{\alpha\beta}]_{cd} dx_2\wedge dx_3 \wedge dx_0 \wedge dx_1\\
+&\frac{1}{8}\int_{\bR^4}\epsilon^{abcd}B^3_\gamma B^1_\mu[\hat{E}^{\gamma \mu}]_{ab} \cdot \partial_0 A^2_{\alpha\beta}[\hat{E}^{\alpha\beta}]_{cd} dx_3\wedge dx_1 \wedge dx_0 \wedge dx_2.
\end{align*}
We sum over repeated indices.

\begin{rem}\label{r.tr.4}
\begin{enumerate}
  \item Using the Minkowski metric $\eta$, we see that the metric $g^{ab} = B_\mu^a\eta^{\mu\gamma}B_\gamma^b$ is in fact degenerate. LQG only applies when we encounter a singularity in the metric. Quantization of the metric is avoided. See \cite{doi:10.1142/S0217732302006692}.
\item Hooft in \cite{THOOFT2001157} also mentioned about applying axial gauge fixing to the Einstein-Hilbert action. By applying axial gauge fixing, we distinguish time from space. Earlier on in Section \ref{s.sn}, we also described the foliation of space, which gives us the notion of a flow of time.
\item The reader may object to apply axial gauge fixing to $e$; after all $e$ is not a gauge and in General Relativity, $e$ defines a metric, which is non-degenerate in classical General Relativity. However, as discussed in \cite{Witten:1988hc}, we must consider $e$ to be non-invertible to make sense of or develop 2+1 quantum gravity. Likewise here, to develop a 3+1 quantum gravity, we have to consider non-invertible $e$.
  \item\label{r.tr.4a} The authors in \cite{0264-9381-21-15-R01} regarded 3-dimensional gravity as a `gauge theory' and thus they considered a Gauss constraint. This constraint is applied to consider connections modulo gauge transformations. See also \cite{Thiemann:2002nj}. This is analogous to the axial gauge fixing applied to $e$ and $\omega$. A similar gauge fixing was also done in \cite{PhysRevD.73.124038}. By considering only gauge invariant forms, the Gauss constraint is automatically satisfied.
\end{enumerate}
\end{rem}

In \cite{PhysRevLett.61.1155}, the authors described the loop representation of quantum gravity in a 3-manifold. The idea is to define a functional on an equivalence class of links, which is invariant under diffeomorphism of the 3-manifold. In other words, this functional should yield the same value for any two links, ambient isotopic to each other, and thus the functionals should produce link invariants.

But Einstein's theory of gravitation is on a 4-manifold. Any set of non-intersecting loops in $\bR \times \bR^3$ is ambient isotopic to a trivial hyperlink. Henceforth, we will consider an equivalence class of such loops, termed time-like isotopic, as described in Section \ref{s.tls}. We will now proceed to define a loop representation of quantum gravity in $\bR^4$, as described in \cite{PhysRevLett.61.1155}.

Start with two hyperlinks $\oL$ and $\uL$, the former termed as a matter hyperlink $\oL = \{\ol^u:\ u=1, \ldots, \on\}$ and the latter $\uL = \{\ul^v:\ v=1, \ldots, \un\}$ as a geometric hyperlink. Orient both hyperlinks and color the components of the matter hyperlink with a representation $\rho$. This means choose a representation $\rho_u: \mathfrak{su}(2) \times \mathfrak{su}(2) \rightarrow {\rm End}(V_u^+) \times {\rm End}(V_u^-)$ for each component loop $\ol^u$, $u=1, \ldots, \on$, in the hyperlink $\oL$. Note that we do not color the components loops in $\uL$, i.e. we do not choose a representation for $\uL$.

The spin connection $\omega$ is $\mathfrak{su}(2) \times \mathfrak{su}(2)$ valued one-form. The first copy of $\mathfrak{su}(2)$ is generated by $\{\hat{E}^{0i}\}_{i=1}^3$, which corresponds to boost in the $x_i$ direction in the Lorentz group; the second copy of $\mathfrak{su}(2)$ is generated by $\{\hat{E}^{\tau(i)}\}_{i=1}^3$, which corresponds to rotation about the $x_i$-axis in the Lorentz group. When we give a representation $\rho^\pm$ to a loop $\ol$, which we interpret as representing a particle, we are effectively assigning values to the translational and angular momentum of this particle.

Given the 2 sets of hyperlinks $\oL \equiv \{\ol^u\}_{u=1}^{\on}$ and $\uL\equiv \{\ul^v\}_{v=1}^{\un}$, there are many (infinite) ways to tangle the 2 hyperlinks together to form another hyperlink with $\on + \un$ components. The reason is because we only consider equivalence classes of hyperlinks, so it does not make sense to just consider the union of $\oL$ with $\uL$. Hence we need and will make a choice as to how they are tangled together, and denote this new hyperlink by $\chi(\oL, \uL) \equiv \chi(\{\ol^u\}_{u=1}^{\on}, \{\ul^v\}_{v=1}^{\un})$. We will also refer it as a colored hyperlink, assumed to be time-like.

Let $q$ be some real constant, known as the charge. Define
\begin{align}
\begin{split}
V(\{\ul^v\}_{v=1}^{\underline{n}})(e) :=& \exp\left[ \sum_{v=1}^{\un} \int_{\ul^v} \sum_{\gamma=0}^3 B^i_\gamma \otimes dx_i\right], \\
W(q; \{\ol^u, \rho_u\}_{u=1}^{\on})(\omega) :=& \prod_{u=1}^{\on}\Tr_{\rho_u}  \mathcal{T} \exp\left[ q\int_{\ol^u} A^i_{\alpha\beta} \otimes dx_i\otimes \hat{E}^{\alpha\beta}  \right].
\end{split}\label{e.l.7}
\end{align}
Here, $\mathcal{T}$ is the time-ordering operator as defined in \cite{CS-Lim02} and $\Tr_{\rho}$ is the matrix trace in the representation $\rho$. We sum over repeated indices, with $i$ taking values in 1, 2 and 3.

A loop representation is an Einstein-Hilbert path integral, of the form \beq \frac{1}{Z}\int_{\omega \in L_\omega,\ e \in L_e}V(\{\ul^v\}_{v=1}^{\un})(e) W(q; \{\ol^u, \rho_u\}_{u=1}^{\on})(\omega)  e^{i S_{{\rm EH}}(e, \omega)}\ De D\omega, \label{e.eh.1} \eeq whereby $De$ and $D\omega$ are Lebesgue measures on $L_e$ and $L_\omega$ respectively and \beq Z = \int_{\omega \in L_\omega,\ e \in L_e}e^{i S_{{\rm EH}}(e, \omega)}\ De D\omega. \label{e.l.8} \eeq Note that it is indexed by a parameter $\kappa > 0$. In \cite{EH-Lim02}, we made sense of this path integral. When we take the limit as $\kappa$ goes to infinity, this sequence of path integrals will define a functional, termed as a Wilson Loop observable of a colored hyperlink $\chi(\oL, \uL)$.

\begin{rem}
A similar expression was also considered in \cite{PhysRevLett.61.1155}, which was argued to give us knot invariants, the knots being embedded in $\bR^3$.
\end{rem}

In \cite{EH-Lim03} or \cite{EH-Lim06}, we defined the hyperlinking number ${\rm sk}(\ol^u, \uL)$ between an oriented time-like loop $\ol^u$ and an oriented time-like hyperlink $\uL$, computed from $\chi(\oL, \uL)$. Suppose the matter hyperlink $\oL$ is colored with  $\{\rho_u\}_{u=1}^{\on} = \{(\rho_u^+, \rho_u^-)\}_{u=1}^{\on}$. Then in \cite{EH-Lim03}, we defined the Wilson loop observable for the colored hyperlink $\chi(\oL, \uL)$ as
\begin{align*}
Z(q; & \chi(\oL, \uL) )
:= \prod_{u=1}^{\on}\left( \Tr_{\rho^+_u}\ \exp[-\pi iq\ {\rm sk}(\ol^u, \uL) \cdot \mathcal{E}^+] +
\Tr_{\rho^-_u}\ \exp[\pi iq\ {\rm sk}(\ol^u, \uL) \cdot \mathcal{E}^-] \right).
\end{align*}
This leads us to make the following definition.

\begin{defn}\label{d.z.1}
Consider the representation $\{(\rho_u^+, 0)\}_{u=1}^{\on}$, whereby $\ol^u$ is colored with the representation $(\rho_u^+, 0)$. Using the setup as above, define Expression \ref{e.eh.1} as
\begin{align*}
Z^{+}(q;  \chi(\oL, \uL) )
&:= \prod_{u=1}^{\on}\ \Tr_{\rho^{+}_u}\ \exp[-\pi iq\ {\rm sk}(\ol^u, \uL) \cdot \mathcal{E}^{+}].
\end{align*}

Now consider the representation $\{(0, \rho_u^-)\}_{u=1}^{\on}$, whereby $\ol^u$ is colored with the representation $(0, \rho_u^-)$. Similarly, we define  Expression \ref{e.eh.1} as
\begin{align*}
Z^{-}(q;  \chi(\oL, \uL) )
&:= \prod_{u=1}^{\on}\ \Tr_{\rho^{-}_u}\ \exp[\pi iq\ {\rm sk}(\ol^u, \uL) \cdot \mathcal{E}^{-}].
\end{align*}

\end{defn}

\begin{rem}
When $\rho_u \equiv (\rho_u^+,0)$, then $Z$ reduces to $Z^+$. Similarly for $\rho_u \equiv (0, \rho_u^-)$, $Z$ reduces to $Z^-$.
\end{rem}

Even though the hyperlinking number is well-defined for a pair of time-like hyperlinks $\{\oL, \uL\}$, tangled to form a time-like hyperlink $\chi(\oL, \uL)$, it is however not an invariant under the equivalence relation given in Definition \ref{d.ts.1}. To obtain an equivalence invariant, it is necessary to time-order the loops, as was done in \cite{EH-Lim06}.

\begin{defn}(Time-ordered pair of hyperlinks)\label{d.tl.2}\\
Suppose we have a time-like hyperlink, denoted as $\chi(\oL, \uL)$, consisting of 2 non-empty sets of time-like hyperlink, $\oL= \{\ol^1, \cdots, \ol^{\overline{n}}\}$ and $\uL= \{\ul^1, \cdots, \ul^{\underline{n}}\}$.

Refer to Definition \ref{d.tau.1}. Pick a component loop $\ol^u \subset \oL$ and another component loop $\ul^v \subset \uL$. We require that either $\ol^u < \ul^v$ or $\ol^u > \ul^v$. In this case, we say that the time-like hyperlink $\chi(\oL, \uL)$ consists of a time-ordered pair of time-like hyperlinks.

Let $u=1, \ldots, \on$. We write $\ol^u < \uL$ if we have $\ol^u < \ul^v$ for every $v = 1, \ldots, \un$; $\ol^u > \uL$ if $\ol^u > \ul^v$ for every $v = 1, \ldots, \un$.
\end{defn}

\begin{rem}\label{r.tr.1}
\begin{enumerate}
  \item Suppose we have a matter loop $\ol$ and a geometric loop $\ul$, both tangled together to form a time-like hyperlink $\chi(\ol,\ul)$. Suppose that $\ol < \ul$. In \cite{EH-Lim06}, we showed that the hyperlinking number ${\rm sk}(\ol, \ul)$ will be equal to $3 \times {\rm lk}(\pi_0(\ol), \pi_0(\ul))$; we will add a negative sign in front if $\ol > \ul$. Here, ${\rm lk}(\pi_0(\ol), \pi_0(\ul))$ is the linking number between knots $\pi_0(\ol)$ and $\pi_0(\ul)$.
  \item Because both $Z^+(q;  \chi(\oL, \uL) )$ and $Z^-(q;  \chi(\oL, \uL) )$ defined earlier, depend on the hyperlinking numbers between $\ol^u$ and $\uL$, we see that it is actually dependent on the time-ordering between each matter loop $\ol^u$ and each component geometric loop in $\uL$.
\end{enumerate}
\end{rem}

\begin{notation}\label{n.l.1}
Let $L(\bR^4)$ denote the space of (oriented) hyperlinks in $\bR^4$. Now the space of hyperlinks is too big for our use. Recall we have the equivalence relation defined in Definition \ref{d.ts.1}, which we will denote by $\sim$. Thus we can consider the quotient space $L(\bR^4)/\sim$.

Create 2 copies of this quotient space $L(\bR^4)/\sim$. The first copy will be denoted by $\oL(\bR^4)$,
termed as the space of matter hyperlinks; the second copy will be denoted by $\uL(\bR^4)$,  termed as the
space of geometric hyperlinks.

By abuse of notation, we will write $\oL$ to denote an equivalence class, using the time-like hyperlink $\oL$ as a representative. Similarly, $\uL$ will also denote an equivalence class, representing all time-like hyperlinks time-like isotopic to $\uL$.

From $\oL(\bR^4)$, we choose a sequence of oriented hyperlinks, which we will denote this sequence by $\{\oL_{m} \}_{m=1}^\infty$. In the sequence $\{\oL_{m} \}_{m=1}^\infty$, for each $\oL_{m} = \{\ol^{1}_m, \ldots, \ol^{\on(m)}_m \}$, we let $\on(m)$ denote the number of components in $\oL_m$ and color each component $\ol^{u}_m$ in the hyperlink $\oL_{m}$ with a representation $\rho_{m, u}:\mathfrak{su}(2)
\times \mathfrak{su}(2) \rightarrow {\rm End}(V_{m,u}^+) \times {\rm End}(V_{m,u}^-)$. It shall be noted that matter hyperlinks will always be colored.

For each matter hyperlink $\oL_{m}$, we choose an oriented geometric hyperlink $\uL_{m}\equiv \{\ul_m^1, \ldots, \ul_m^{\un(m)}\} \in \uL(\bR^4)$, and decide how they are tangled together, denoted by $\chi(\oL_m, \uL_m)$. Hence, \beq
\chi(\oL_m, \uL_m) = \{\ol^1, \ldots, \ol^{\on(m)}, \ul^1, \ldots, \ul^{\un(m)}\} \nonumber \eeq
is a new colored time-like oriented hyperlink formed from $\oL_m = \{\ol^1, \ldots, \ol^{\on(m)}\}$ and $\uL = \{\ul^1, \ldots, \ul^{\un(m)}\}$, with each component loop of $\oL_m$ tangled with $\uL_m$, represented by a time-like equivalence class, whereby each $\ol^k \in \oL_m$ is time-ordered with $\ul^r \in \uL_{m}$, for each $r=1,2, \cdots, \un(m)$, as described in Definition \ref{d.tl.2}. Note that $\un(m)$ is the number of components in $\uL_m$.

We will call the sequence $\{\uL_{m}\}_{m=1}^\infty$ as the configuration space.
\end{notation}

\begin{rem}\label{r.tr.3}
\begin{enumerate}
  \item The sequence $\{ \oL_{m}\}_{m=1}^\infty$ represents all the possible quantum states of the particles. 
  \item Recall we have $\oL$ and $\uL$ as two separate oriented hyperlinks. Note that $\oL$ should represent a set of particles in $\bR^4$, and from QFT, we have a discrete number of possible particle states. We will also interpret the projected component matter loop inside $\bR^3$, as representing the orbit of a particle. The other hyperlink, $\uL$, should be thought of as the dual to the hyperlink $\oL$.
  \item When we write $\chi(\oL_m, \uL_m)$, it is understood that it means an equivalence class $[\chi(\oL_m, \uL_m)]$. We say that $\chi(\oL_m, \uL_m)$ is equivalent to $\chi(\oL_m^\ast, \uL_m^\ast)$ if
      \begin{itemize}
        \item $\chi(\oL_m, \uL_m)$ is time-like isotopic to $\chi(\oL_m^\ast, \uL_m^\ast)$;
        \item during the time-like isotopic process, the time-ordering between each component loop in the matter hyperlink and each component loop in the geometric hyperlink is unchanged.
      \end{itemize}
  \item\label{r.tr.3a} By considering equivalence classes of time-like isotopic hyperlinks, we are in fact imposing the spatial diffeomorphism constraint.
  \item\label{r.tr.3b} For some foliation of space, we can always consider the matter loop as embedded inside $\bigcup_{\epsilon_0< t < \epsilon_1}\Sigma_t$, and the geometric loop as embedded inside $\bigcup_{\delta_0< t < \delta_1}\Sigma_t$, whereby either $\epsilon_1 < \delta_0$, or $\delta_1 < \epsilon_0$. This gives us a time-ordering between the two loops, which gives rise to causality, as we will discuss in greater detail in the rest of the article. By imposing time-ordering on the matter and geometric loops, we are in fact imposing the Hamiltonian constraint.
\end{enumerate}

\end{rem}


Represent each oriented hyperlink $\oL_{n}$ as a function $\oPsi_{n}^{\ \pm} \equiv \oPsi_{\oL_n}^{\ \pm}$, whose domain is the configuration space, denoted by $\{\uL_{m}\}_{m=1}^\infty$, whereby for each given $m$, $\uL_{m} = \{\ul_m^{v}\}_{v=1}^{\underline{n}(m)}$, $\un(m) \in \mathbb{N}$ is an oriented hyperlink. The function $\oPsi_{n}^{\ \pm}$ at $\uL_{m}$ is defined as \beq \oPsi_{n}^{\ \pm}(\uL_{m}) :=
\left\{
  \begin{array}{ll}
    Z^\pm(q;  \chi(\oL_{n}, \uL_{n})), & \hbox{$m=n$;} \\
    0, & \hbox{$m \neq n$,}
  \end{array}
\right. \nonumber \eeq $Z^\pm(q; \cdot)$ was defined in Definition \ref{d.z.1}. In our previous work \cite{EH-Lim06}, we showed that the hyperlinking number is an equivalence invariant for a time-like oriented hyperlink $\chi(\oL_{n}, \uL_{n})$, as long as during the time-like isotopy, the time-ordering and orientation is preserved. Thus, $Z^\pm(q; \cdot)$ is well-defined. See also Remark \ref{r.tr.1}.

Written as an Einstein-Hilbert path integral, we have
\begin{align*}
\oPsi_{m}^{\ +}(\uL_{m}) :=& \lim_{\kappa \rightarrow \infty}\int_{\omega \in L_\omega,\ e \in L_e}V(\uL_{m})(e) W(q; \{\ol_m^u, (\rho_{m,u}^{+}, 0)\}_{u=1}^{\on(m)})(\omega)  e^{i S_{{\rm EH}}(e, \omega)}\ De D\omega, \\
\oPsi_{m}^{\ -}(\uL_{m}) :=& \lim_{\kappa \rightarrow \infty}\int_{\omega \in L_\omega,\ e \in L_e}V(\uL_{m})(e) W(q; \{\ol_m^u, (0,\rho_{m,u}^{-} )\}_{u=1}^{\on(m)})(\omega)  e^{i S_{{\rm EH}}(e, \omega)}\ De D\omega,
\end{align*}
computed using the colored hyperlink $\chi(\oL_{m}, \uL_{m})$. Note that $L_\omega$ and $L_e$ were defined by Equations (\ref{e.o.1}), and dependent on $\kappa$.

On the flip side, we can consider $\{\oL_{n}\}_{n=1}^\infty$ as the dual configuration space. And we represent each oriented hyperlink $\uL_{m}$ now as a function $\uPsi_{m}^{\ \pm} \equiv \uPsi_{\uL_m}^{\ \pm}$, acting on this dual space. It is no surprise that we can define this functional as \beq \uPsi_{m}^{\ \pm}(\oL_{n}) :=
\left\{
  \begin{array}{ll}
    Z^\pm(q;  \chi(\oL_{m}, \uL_{m})), & \hbox{$n = m$;} \\
    0, & \hbox{$n\neq m$.}
  \end{array}
\right.
 \nonumber \eeq

\begin{defn}(Phase space)\\
Let $\{(\uL_{m}, \oL_m)\}_{m=1}^\infty$ be known as the phase space.

Let $\oV^{\ \pm}$ be the vector space, spanned by the basis $\{\oPsi_{n}^{\ \pm}\}_{n=1}^\infty$. This will be the vector space of real-valued functionals acting on the configuration space. 

Let $\uV^{\ \pm}$ be the vector space, spanned by the basis $\{\uPsi_{m}^{\ \pm}\}_{m=1}^\infty$. This will be the vector space of real-valued functionals acting on the dual configuration space. 

Let $V^\pm$ be the subspace inside $\oV^{\ \pm} \otimes \uV^{\ \pm}$, spanned by the basis $\{\oPsi_m^{\ \pm} \otimes \uPsi_{m}^{\ \pm}\}_{m=1}^\infty$. This is a vector space of functionals acting on the phase space $\{(\uL_{m}, \oL_m)\}_{m=1}^\infty$, by \beq \left(\oPsi_n^{\ \pm} \otimes \uPsi_{n}^{\ \pm} \right)(\uL_m, \oL_m) = \oPsi_n^{\ \pm}(\uL_m) \uPsi_{n}^{\ \pm}(\oL_m). \nonumber \eeq
\end{defn}

\begin{defn}(Construction of Quantum Loop Space)\label{d.qls}\\
The set $\oL := \{\oL_{m}\}_{m=1}^\infty$ is a discrete space, and we will equip it with a probability counting measure, denoted by $N: \mathbb{N} \rightarrow [0,1]$, such that $\sum_{m=1}^\infty N(m) = 1$.

Define an inner product on $\oV^{\ \pm}$ by \beq \left\langle \oPsi_{m}^{\ \pm}, \oPsi_{\bar{m}}^{\ \pm} \right\rangle
:= \sum_{n=1}^\infty \oPsi_{m}^{\ \pm}(\uL_{n})\oPsi_{\bar{m}}^{\ \pm}(\uL_{n})N(n). \nonumber \eeq Complete $\oV^{\ \pm}$ into a Hilbert space $\mathcal{H}(\oV^{\ \pm})$.

Similarly, complete $\uV^{\ \pm}$ into a Hilbert space $\mathcal{H}(\uV^{\ \pm})$ with the inner product \beq \left\langle \uPsi_{m}^{\ \pm}, \uPsi_{\bar{m}}^{\ \pm} \right\rangle
:= \sum_{n=1}^\infty \uPsi_{m}^{\ \pm}(\oL_{n})\uPsi_{\bar{m}}^{\ \pm}(\oL_{n})N(n). \nonumber \eeq

Finally, complete $V^\pm$ into a Hilbert space  $\mathcal{H}(V^\pm)$ using the inner product \beq \left\langle \oPsi_m^{\ \pm} \otimes \uPsi_{m}^{\ \pm}, \oPsi_{\bar{m}}^{\ \pm} \otimes \uPsi_{\bar{m}}^{\ \pm} \right\rangle := \sum_{n=1}^\infty \oPsi_m^{\ \pm}(\uL_n)\oPsi_{\bar{m}}^{\ \pm}(\uL_n)\uPsi_m^{\ \pm}(\oL_n)\uPsi_{\bar{m}}^{\ \pm}(\oL_n)N(n). \nonumber \eeq
\end{defn}

\begin{rem}
\begin{enumerate}
  \item The reader should think of $N(m)$ as the probability for a set of particles to be represented by $\oL_m$. Note that the number of particles is not conserved in the dual configuration space.
  \item In the construction, we see that both $\{\oPsi_{n}^{\ \pm}\}_{n=1}^\infty$ and $\{\uPsi_{m}^{\ \pm}\}_{m=1}^\infty$ are bases for the Hilbert spaces $\mathcal{H}(\oV^{\ \pm})$ and $\mathcal{H}(\uV^{\ \pm})$ respectively. Thus, all the Hilbert spaces we have constructed are separable. By their construction, both are orthogonal bases for the chosen counting measure.
  \item The reason why we made the bases orthogonal is because later on, we will show that they will form an eigenbasis for the operators we are interested in. Because operators corresponding to observables must be Hermitian operators by a postulate in QM, hence we made the bases orthogonal.
\end{enumerate}

\end{rem}

In  QM, one would define a time-evolution operator, which is unitary, hence the probability is conserved as the states evolve with time. But in LQG, there is no absolute notion of time, so one should not attempt to define a time-evolution operator. One way to remedy this situation, as proposed by Hooft, would be to define an equivalence relation, whereby two states are deemed to be equivalent if they evolve into a same state in the future. One then constructs a Hilbert space, whereby each element in a basis is identified with such an equivalence class. For details, refer to  \cite{THOOFT2001157}.

In our construction, time-like equivalence and time-ordering ensure that a hyperlink $\chi(\oL, \uL)$, will be equivalent to the same hyperlink, be it translated forward or backward in time. This means that our equivalence relation implies the relation suggested by Hooft, but not conversely. And the preceding probability measure defined on the equivalence class $[\oL_{m}]$ or $[\uL_{m}]$, represented by a hyperlink $\oL_{m}$ or $\uL_m$ respectively, will not change under translation in time.

In the absence of a space-time metric, time and space are ambiguous. However, from our definition of the hyperlinking number, one sees that time is distinguished from space. Earlier on, we mentioned that the hyperlinking number will be an invariant, if we consider an equivalence class of a pair of time-like hyperlinks $\chi(\oL, \uL)$, with the additional condition that each component matter loop is time-ordered with a component geometric loop. One should then interpret causality as a consequence of time-ordering. Thus, we only consider diffeomorphisms of space-time that preserves causality. In a topological theory on $\bR^4$, there is no distinguished axis that is labelled as time, hence there will be no notion of causality. Considering time-like and time-ordered pair of loops, is not a flaw in the LQG theory, but rather it preserves the notion of time as an external parameter in QM, and causality in Special Relativity. And despite this restriction, it still keeps the malleability of space-time, as dictated in GR.

\section{Observables}\label{s.obs}

In the previous section, we constructed the phase space. Notice that we used (time-like) hyperlinks to construct the phase space, which are geometrical objects. In loop quantum gravity, there is no preferred metric, so we need to consider equivalence classes of a hyperlink. Hence, when we choose two equivalence classes of hyperlink $\oL$ and $\uL$, we need to make a choice as to how they are tangled. Furthermore, we time-ordered a pair of matter and geometric component loops in Definition \ref{d.tl.2}.

If a matter loop occurs before a geometric loop, one can interpret physically as saying that the matter loop $\ol$ will produce a geometric loop $\ul$, that tangles with it. Furthermore, each matter hyperlink $\oL$ can only produce one geometric hyperlink, up to time-like equivalence. The tangled hyperlink $\chi(\oL, \uL)$ is an equivalence class of hyperlink, consisting of time-ordered pair of hyperlinks. By definition of the equivalence class, if a matter loop precedes before a geometric loop, then any ambient isotopy of these loops must maintain that the transformed matter loop is the cause, the transformed geometric loop is the effect. Thus, causality is preserved under the equivalence relation.

In this section, we would like to discuss the observables that one can define on the phase space. Interestingly enough and to no surprises, the observables we are going to consider will be geometrical objects.

\begin{defn}\label{d.s.1}
Let $l$ be a time-like loop and $S = \bigcup_{k=1}^n S_k$, whereby $S_k$ is a connected compact surface with boundary $\partial S_k$. We say that $l$ is time-ordered with $S$, if $\tau(\vec{x}) < \tau(\vec{y})$ or $\tau(\vec{x}) > \tau(\vec{y})$, for every $\vec{x} \in l$ and $\vec{y} \in S_k$. We will denote this relation as $l < S_k$ ($l > S_k$) for the former (latter), by abuse of notation.
\end{defn}

On the configuration space $\{\uL_m\}_{m=1}^\infty$, choose an orientable compact surface, $\overline{S} \subset \{0\} \times \bR^3$, thus $\overline{S}$ is a Seifert surface for $\partial\overline{S}$, if $\overline{S}$ has non-empty boundary. The surface $\overline{S}$ is represented up to ambient isotopy in $\bR \times \bR^3$ and disjoint from the geometric hyperlinks in the configuration space. We also assume that its boundary $\partial \overline{S}$, if any, is a time-like hyperlink.

Pick a geometric hyperlink $\uL_m$ from the configuration space, disjoint from $\overline{S}$, and we need to make a choice as to how $\overline{S}$ and $\uL_m$ are `linked together'. Each component geometric loop must be time-ordered with each connected component surface with boundary in $\overline{S}$. If $\overline{S}$ has no boundary, then time-ordering is not required.

Write $\{\overline{S}, \emptyset , \chi(\oL_m,\uL_m)\}$ to denote a time-like triple, as defined in the last section of \cite{EH-Lim06}. If we orientate the hyperlink and the surface, then this oriented time-like triple represents the equivalence class of these oriented geometrical objects in $\bR \times \bR^3$. On this time-like oriented triple, the linking number between $\overline{S}$ and $\uL_m$, denoted as ${\rm lk}(\uL_m, \overline{S})$, as defined in \cite{EH-Lim06}, and the hyperlinking number between $\oL_m$ and $\uL_m$, are time-like equivalence invariants, preserving the time-ordering.

A piercing is formed when an arc in the projected loop $\pi_0(l)$ intersects the projected surface $\pi_0(S)$, without being tangent to the surface. The linking number is then computed by adding up all the algebraic numbers of these piercings.

In \cite{EH-Lim06}, we showed that when the surface has no boundary, then the linking number between the oriented closed surface and an oriented loop is invariant under ambient isotopy. But, if the surface has a boundary, then the linking number between the oriented surface and the oriented loop, is not a topological invariant. Under the above equivalence relation, a time-ordering needs to be imposed and suppose $l < S$ for a loop $l$ and a connected surface $S$. Then, we only consider ambient isotopy whereby during the ambient isotopic process, time-ordering remains unchanged and it further shows that
\begin{itemize}
  \item the time-like oriented hyperlink formed from $l$ and $\partial S$, is time-like isotopic to the time-like oriented hyperlink formed from $l'$ and $\partial S'$;
  \item $S$ is ambient isotopic to $S'$.
\end{itemize} 
Thus, causality is respected under the equivalence relation.

When $S$ has no boundary, no time-ordering is required under the equivalence relation. But, we saw in \cite{EH-Lim06} that indeed a time-ordering is implicitly implied and preserved, under ambient isotopy of the loop and a connected surface. When a piercing is formed before (after) the time duration for which the surface is formed, we will refer it as a left (right) piercing. Using ambient isotopy, this implicitly defined time-ordering looks at whether all projected oriented arcs going from a left to a right piercings, lie either in the interior or exterior of $\pi_0(S) = S \subset \bR^3$, during the time duration for which the connected surface is formed. By simply reversing the orientation of the loop, we can change this time-ordering. Thus, the linking number between a connected closed surface and a loop is dependent on this implicitly defined time-ordering. This time-ordering is of course preserved during any ambient isotopy.

On the dual configuration space $\{\oL_n\}_{n=1}^\infty$, we will again choose 2 geometrical objects as observables; an orientable surface $\underline{S} \subset \bR^3$ and a compact (3-dimensional) region $R \subset \bR^3$. We will identify both geometrical objects inside $\bR \times \bR^3$ by $\{0\} \times \underline{S}$ and $\{0\} \times R$ respectively and note that we choose representatives for this surface and solid region, up to ambient isotopy in $\{0\} \times \bR^3$ and $\bR^3$ respectively.

Each matter hyperlink $\oL_m$ from the dual configuration space must be disjoint from $\underline{S}$ and $R$. For a matter hyperlink, we need to choose a frame for each component loop in $\oL_m$, so that $\pi_0(\oL_m)$ is a framed link in $\bR^3$. The framing on the link $\pi_0(\oL_m)$ will give rise to nodes on it. See \cite{EH-Lim06}.

We now need to make a choice as to how the given oriented surface $\underline{S}$ and the oriented matter hyperlink $\oL_m$ are linked. Again, time-ordering needs to be imposed, when $\underline{S}$ has an non-empty boundary, assumed to be a time-like hyperlink. As for $R$, we need to decide how many nodes on $\pi_0(\oL_m)$ lie inside $R$.

Once decided, form an oriented time-like triple $\{\underline{S}, R, \chi(\oL_m, \uL_m)\}$, as defined in \cite{EH-Lim06}. Note that $\underline{S}$ and $R$ are represented as an ambient isotopic class in $\bR^3$. This oriented time-like triple represents an equivalence class of these geometrical objects. The piercing number $\nu_{\underline{S}}(\oL_m)$ counts the number of piercings formed between $\pi_0(\oL_m)$ and $\pi_0(\underline{S})$. And the confinement number, denoted by $\nu_R(\oL_m)$, counts how many nodes on $\pi_0(\oL_m)$, lie in the interior of $R$. On this triple, the piercing and confinement numbers are time-like, preserving time-ordering equivalence invariants. For more details, refer to \cite{EH-Lim06}.

\begin{rem}\label{r.tr.2}
Let $\vec{x} = (x_0, x)$ be a point on a loop in $\bR^4$, which projects down to a node $\pi_0(\vec{x})$ on a framed knot in $\bR^3$. Assume that this node lies inside the region $R$. Because the loop does not intersect with $\{0\} \times R$, we have either $x_0 > 0$ or $x_0 < 0$. Suppose $x_0 > 0$. Under the equivalence relation, when the framed loop is transformed under time-like isotopy, then $\vec{x}$ will be mapped to a point $\vec{y}$, such that
\begin{itemize}
  \item $\pi_0(\vec{y})$ will be a node on the transformed knot, and it must remain in the interior of the region $R$;
  \item during the time-like isotopy, any transformed loop must remain disjoint from the region $R$ and nodes on projection, must always remain disjoint from the boundary of $R$.
\end{itemize}
Thus, we also must have $y_0 > 0$. In other words, those nodes on the projected loop with time component positive, will be mapped to nodes on the transformed knot, with time component positive. Therefore, one can say that the nodes are implicitly time-ordered and under the equivalence relation, this time-ordering does not change, again implying that causality is preserved.
\end{rem}

Our description of the geometrical objects in consideration is now complete and we remind the reader, that they are represented, inside $\bR \times \bR^3$, up to time-like isotopy, preserving time-ordering. We will introduce physical quantities that can be computed using the above described oriented time-like triple. Because there is no preferred metric, the quantities that one is able to compute from these geometrical objects must be invariants under the equivalence relation.

In the next section, we will introduce observables, which will in turn define quantum operators, whose eigenvalues are computed using these invariants.

\section{Quantized operators}\label{s.qo}

To develop a quantum theory of gravity, we need to quantize certain observables. Area and volume are observables that depend on the vierbein $e$; the curvature depends on the spin connection $\omega$.

We need to find a basis for $\mathcal{H}(\oV^{\ \pm})$, such that they form an eigenbasis for the area operator $\hat{A}$ and for the volume operator $\hat{V}$. Then, the respective eigenvalues will be the quantized values for area and volume. Similarly, we need to find a basis for $\mathcal{H}(\uV^{\ \pm})$, such that they form an eigenbasis for the quantized curvature operator $\hat{F}$, and its eigenvalues will be the quantized values for curvature.

Recall in Section \ref{s.cvs}, we wrote down Expression \ref{e.eh.1} using the functionals $V$ and $W$. We can now proceed to quantize the observables using the Einstein-Hilbert path integral.

\subsection{Area operator}\label{ss.ao}

Let $S$ be an orientable compact surface inside the spatial subspace $\bR^3 \hookrightarrow \bR \times \bR^3$, disjoint from the matter hyperlink. Because we can consider ambient isotopy of $S$ in $\bR^3$, we assume that $S$ is inside $x_2-x_3$ plane. Furthermore, we insist that $\pi_0(\oL)$ intersect the surface $S$ at most finitely many points. Using the dynamical variables $\{B_\mu^a\}$ and the Minkowski metric $\eta^{ab}$, we see that the metric $g^{ab} \equiv B^a_\mu\eta^{\mu\gamma}B^b_\gamma$ and the
area is given by \beq {\rm Area\ of}\ S(e) := A_S(e) \equiv \int_S \sqrt{g^{22}g^{33} - (g^{23})^2} \ dA. \nonumber \eeq

In \cite{EH-Lim03}, we quantized the area of $S$ into an operator $\hat{A}_S$ using the following path integral expression (indexed by a parameter $\kappa$)
\beq \frac{1}{Z}\int_{\omega \in L_\omega,\ e \in L_e}V(\{\ul^v\}_{v=1}^{\un})(e) W(q; \{\ol^u, \rho_u\}_{u=1}^{\on})(\omega)A_S(e)\  e^{i S_{{\rm EH}}(e, \omega)}\ De D\omega, \label{e.eha.1} \eeq $V$ and $W$ as defined in Equation (\ref{e.l.7}) and $Z$ is a normalization constant given by Equation (\ref{e.l.8}). Note that $\{\rho_u^\pm\}_{u=1}^{\on}$ is a representation for the hyperlink $\oL$.

The main theorem in \cite{EH-Lim03} says that the limit of this sequence of path integrals in Expression \ref{e.eha.1} as $\kappa$ goes to infinity, can be computed from the number of times the projected hyperlink $\pi_0(\oL)$ pierce the surface $S$, weighted by the momentum corresponding to the representation $\rho_u$. When $\rho_u = \rho_u^\pm$, we will write this limit as $\hat{A}_S[Z^\pm(q; \chi(\oL, \uL))]$, for which the reader may recall, we defined $Z^\pm(q; \chi(\oL, \uL))$ in Definition \ref{d.z.1}.

Recall in Section \ref{s.cvs}, for a colored hyperlink $\oL$, we can define a functional $\oPsi_{\oL}^{\ \pm}$ acting on the configuration space. In \cite{EH-Lim06}, we defined the piercing number $\nu_S(l)$ between a compact surface $S$ in $\{0\} \times \bR^3$, with or without boundary, and a loop $l$. This is a well-defined invariant, up to time-like isotopy, preserving time-ordering. It counts the number of times $\pi_0(l)$ intersects the surface $S$ in $\bR^3$, the intersection points termed as piercings. We always choose a representative of $l$ and $S$ such that the $\pi_0(l)$ and $S$ in $\bR^3$ have the minimum number of piercings, which gives us $\nu_S(l)$.

We have the following corollary from the main theorem in \cite{EH-Lim03}.

\begin{cor}\label{c.a.1}
Given an orientable compact surface $S \subset \{0\} \times \bR^3 $, quantization of the area of $S$, gives an operator $\hat{A}_S$, which  acts on functionals in $\oV^{\ \pm}$.

Suppose $\oPsi_{\oL}^{\ \pm} \in \oV^{\ \pm}$, with $\oL = \{\ol^u\}_{u=1}^{\on}$. Then \beq \hat{A}_{S} \oPsi_{\oL}^{\ +} = \frac{|q|\sqrt\pi}{2} \left[
\sum_{u=1}^{\on}\nu_{S}(\ol^u)  \sqrt{\xi_{\rho_u^+}}\right]\oPsi_{\oL}^{\ +} \nonumber \eeq and \beq \hat{A}_S \oPsi_{\oL}^{\ -} = i\frac{|q|\sqrt\pi}{2} \left[
\sum_{u=1}^{\on}\nu_{S}(\ol^u)\sqrt{\xi_{\rho_u^-}}\right]\oPsi_{\oL}^{\ -}. \nonumber \eeq

Hence, the canonical basis $\{\oPsi_{n}^{\ \pm}\}_{n=1}^\infty$ is an eigenbasis for $\hat{A}_S$. Furthermore, the basis is orthogonal using the counting measure $N$.
\end{cor}

\begin{rem}
The eigenvalues are either purely real or purely imaginary.
\end{rem}

Note that $\nu_S(\ol^u)$ is computed from the oriented time-like triple $\{S, \emptyset, \chi(\oL,\uL)\}$ and the eigenvalues of $\hat{A}_S$ are invariant under diffeomorphism of $\bR \times \bR^3$, provided it respects time-like isotopy and also preserves time-ordering.

\begin{rem}\label{r.to.1}
Time-ordering between $S$ and $\oL$ is only necessary when $S$ has a boundary. When a connected surface $S$ has no boundary, a time-ordering is implicitly defined, as discussed in Section \ref{s.obs}. In the event that we have to time-order the surface $S$ with the component matter loops, note that the piercing number is independent of any time-ordering, as it only counts the total number of piercings, without taking into account of their algebraic sign. So the Area operator is actually independent of any time-ordering imposed or implicitly defined.
\end{rem}

The term $\sqrt{\xi_\rho}$ appears in the expression and it is interpreted as total momentum for representation $\rho$. So, the eigenvalues of the area operator gives us the total momentum impact by the matter hyperlink on a surface $S$, taking into account momentum coming from either boost or rotation. As an application of this result, we will derive the Bekenstein-Hawking expression for the entropy of a Schwarzschild black hole in Section \ref{s.bh}.

In the derivation of the path integral for the area, it was necessary to write the surface $S$ as a finite disjoint union of smaller surfaces. In \cite{EH-Lim03}, we showed that the main result was independent of how we partition the surface. Suppose we write $S = \bigcup_{m} S_m$. As the result says that the area is given by counting the number of piercings between $\pi_0(l)$ and $S$, we see that each smaller surface $S_m$, either has one quanta of area, or zero quanta of area, which corresponds to the projected loop piercing the surface once, or there is no piercing respectively. This means that if $S_m$ contains one piercing, any local surface in $S_m$ containing the piercing, will have the same quanta of area as $S_m$.

\subsection{Volume operator}\label{ss.vo}

Fix a closed and bounded 3-manifold $R \subset \bR^3$, possibly disconnected with finite number of components, disjoint from the matter hyperlink.  Henceforth, we will refer to $R$ as a compact solid region. Its boundary is a closed (compact without boundary) surface. Using the dynamical variables $\{B_\mu^a\}$ and the Minkowski metric $\eta^{ab}$, we see that the metric $g^{ab} \equiv B^a_\mu\eta^{\mu\gamma}B^b_\gamma$ and the volume $V_R$ is given by \beq V_R(e) := \int_R \sqrt{\epsilon_{ijk}\epsilon_{\bar{i}\bar{j}\bar{k}}g^{i\bar{i}}g^{j\bar{j}}g^{k\bar{k}}}. \nonumber \eeq

In \cite{EH-Lim04}, we quantized the volume of $R$ into an operator $\hat{V}_R$ using the following path integral expression (indexed by a parameter $\kappa$) \beq \frac{1}{Z}\int_{\omega \in L_\omega,\ e \in L_e}V(\{\ul^v\}_{v=1}^{\un})(e) W(q; \{\ol^u, \rho_u\}_{u=1}^{\on})(\omega)V_R(e)\  e^{i S_{{\rm EH}}(e, \omega)}\ De D\omega, \label{e.v.11} \eeq whereby $V$ and $W$ were defined in Equation (\ref{e.l.7}) and $Z$ is a normalization constant given by Equation (\ref{e.l.8}).

A framed hyperlink when projected in $\bR^3$ using $\pi_0: \bR \times \bR^3 \rightarrow \bR^3$, gives us a framed link. Imagine adding nodes to each component knot inside the framed link. Each node has an algebraic number $\{\pm1\}$ assigned to it and we always assume that all the nodes on a component knot have the same sign. Equivalence class of framed loop $l$ allows us to move the set of nodes in the knot $\pi_0(l)$ and we denote this equivalence class of finite points by ${\rm Nd}(\pi_0(l))$.

The main theorem in \cite{EH-Lim04} says that the limit of the sequence of path integrals in Expression \ref{e.v.11} as $\kappa$ goes to infinity, is computed by adding up all the nodes on a projected framed hyperlink, which lie in the interior of $R$, weighted by the Casimir operator corresponding to the representation $\rho_u$. If we choose the representation $\{\rho_u^\pm\}_{u=1}^{\on}$ for the hyperlink $\oL$, then we will write this limit as $\hat{V}_R[Z^\pm(q; \chi(\oL, \uL))]$, $Z^\pm(q; \chi(\oL, \uL))$ was defined in Definition \ref{d.z.1}.

In \cite{EH-Lim06}, we defined the confinement number $\nu_R(l)$, which counts the number of nodes in ${\rm Nd}(\pi_0(l))$ which are in the interior of $R$. Under the time-like isotopy, as defined in \cite{EH-Lim06}, $\nu_R(l)$ is an invariant.

Recall for a colored hyperlink $\oL$, we can define a functional $\oPsi_{\oL}^{\ \pm}$ acting on a configuration space. As a consequence, we have the following corollary from the main theorem in \cite{EH-Lim04}.

\begin{cor}\label{c.vo.1}
Let $R\subset \bR^3$ be a compact solid region, possibly disconnected with finite components and let $\partial R$ denote its boundary, which is closed. Quantization of the volume of the region $R$, will give us an operator $\hat{V}_R$, which acts on functionals in $\oV^{\ \pm}$.

Suppose $\oPsi_{\oL}^{\ \pm} \in \oV^{\ \pm}$. Then, \beq \hat{V}_R \oPsi_{\oL}^{\ +} = \frac{q^2\pi^{3/2}}{2}\left[ \sum_{u=1}^{\on} \nu_R(\ol^u)
\xi_{\rho_u^+}  \right]\oPsi_{\oL}^{\ +} \nonumber \eeq and \beq \hat{V}_R \oPsi_{\oL}^{\ -} = \frac{q^2\pi^{3/2}}{2}\left[\sum_{u=1}^{\on} \nu_R(\ol^u)
\xi_{\rho_u^-} \right]\oPsi_{\oL}^{\ -}. \nonumber \eeq

Hence, the canonical basis $\{\oPsi_{n}^{\ \pm}\}_{n=1}^\infty$ is an eigenbasis for $\hat{V}_R$. Furthermore, the basis is orthogonal using the counting measure $N$.
\end{cor}

Note that $\nu_R(\ol^u)$ is computed from the oriented time-like triple $\{\emptyset, R, \chi(\oL, \uL)\}$ and the eigenvalues of $\hat{V}_R$ are invariant under diffeomorphism of $\bR \times \bR^3$, provided it respects time-like isotopy, preserving time-ordering.

\begin{rem}\label{r.to.2}
Remark \ref{r.tr.2} said that each node $\pi_0(\vec{x})$ has a definite sign from its time-component $\tau(\vec{x})$, but it does not affect the confinement number. Thus, the confinement number and hence the eigenvalues of the volume operator are actually independent of any time-ordering implicitly defined.
\end{rem}

If one looks at the eigenvalues of $\hat{V}_R$, notice that we have terms involving the square of the total momentum $\sqrt{\xi_{\rho_u^\pm}}$, which in classical mechanics, is interpreted as the kinetic energy of the particle represented by the colored hyperlink $\ol^u$, with representation $\rho_u^\pm$. The Casimir operator $\xi_{\rho_u^+}$ ($\xi_{\rho_u^-}$) will represent translational (rotational) kinetic energy. So, one can view the volume operator as measuring the total kinetic energy of the matter hyperlinks, in a region $R$. Volume operator can be interpreted as energy operator is also evident in \cite{Thiemann:2002nj}.

We can use the eigenvalues of the volume operator to resolve some inconsistency in QFT, as raised by Thiemann in \cite{Thiemann:2002nj}. A particle's momentum $p$ is inversely proportional to its Compton length and its energy $E$ is proportional to its Schwarzschild radius. When its Compton length is comparable to its Schwarzschild radius, then GR predicts that this particle will turn into a black hole, implying that Hawking radiation and all sorts of virtual particles with large momentum and energy will henceforth be emitted.

This will happen at short distances, in the order of Planck's distance (around $10^{-33}$ cm). At such short distances, quantum gravity should come into play.  The main theorem in \cite{EH-Lim04} says that the (kinetic) energy of the particle is discretized. This means that QFT and GR are no longer applicable and the above qualitative picture does not apply.

In the derivation of the volume path integral, it was necessary to write the region $R = \bigcup_{m} R_m$ as a disjoint union of smaller regions. The computations in \cite{EH-Lim04} showed that the final answer is independent of this partition. As the volume of $R$ is given by counting the number of nodes in its interior, we may assume that $R_m$ has either zero or one quanta of volume, which is equivalent to $R_m$ containing zero or one node in its interior respectively. Thus, any smaller region $R' \subset R_m$, will have the same quanta of volume, as long as it contains the node, if any. In other words, any local region inside $R_m$ will contain the same quanta of volume.

\subsection{Curvature operator}\label{ss.co}

Choose an orientable, compact surface $S \subset \bR \times \bR^3$, with or without boundary, which is disjoint from the geometric hyperlink $\uL$. If it has no boundary, we term $S$ as closed. Furthermore, we insist the projected hyperlink $\pi_a(\uL)$ intersect $\pi_a(S)$ at most finitely many points, $a = 0, 1, 2, 3$. If $S$ has a boundary $C$, then $C$ tangles with $\uL$ to form a time-like hyperlink. And we will also time-order the geometric loop with the surface $S$ if there is a boundary, according to Definition \ref{d.s.1}.

Define
\begin{align*}
F_S(\omega)
:= \frac{1}{2}\int_S &\frac{\partial A_{\alpha\beta}^i}{\partial x_0} \otimes dx_0 \wedge dx_i \otimes \hat{E}^{\alpha\beta} + \frac{\partial A_{\alpha \beta}^j}{\partial x_i} \otimes dx_i \wedge dx_j \otimes \hat{E}^{\alpha \beta} \\
& + A_{\alpha\beta}^iA_{\gamma\mu}^j\otimes dx_i \wedge dx_j \otimes [\hat{E}^{\alpha\beta}, \hat{E}^{\gamma\mu}].
\end{align*}

In \cite{EH-Lim05}, we quantized the curvature of $S$ into an operator $\hat{F}_S$ using the following path integral expression (indexed by a parameter $\kappa$) \beq \frac{1}{Z}\int_{\omega \in L_\omega,\ e \in L_e}F_S(\omega) \cdot V(\{\ul^v\}_{v=1}^{\un})(e) W(q; \{\ol^u, \rho_u\}_{u=1}^{\on})(\omega)  e^{i S_{{\rm EH}}(e, \omega)}\ De D\omega, \label{e.ehc.1} \eeq whereby $V$ and $W$ were defined in Equation (\ref{e.l.7}) and $Z$ is a normalization constant given by Equation (\ref{e.l.8}).

\begin{rem}\label{r.c.1}
\begin{enumerate}
  \item A quantized curvature operator at each point $p \in \bR^4$ is not defined, as commented in \cite{PhysRevD.73.124038}.
  \item If we regard the curvature as a field, then it was explained in \cite{streater} that the curvature $R$ do not yield a well-defined operator; rather only smeared fields will yield quantized operators. In this case, we smear the curvature over a surface $S$, given by $F_S$.
\end{enumerate}
\end{rem}

The main theorem in \cite{EH-Lim05} says that the limit of this sequence of path integrals in Expression \ref{e.ehc.1} as $\kappa$ goes to infinity, is computed from the linking number between $\uL$ and the surface $S$, denoted as ${\rm lk}(\uL, S)$. See \cite{EH-Lim06}. If we choose the representation $\{\rho_u^\pm\}_{u=1}^{\on}$ for the hyperlink $\oL$, then we will write this limit as $\hat{F}_S[Z^\pm(q; \chi(\oL, \uL))]$, $Z^\pm(q; \chi(\oL, \uL))$ given in Definition \ref{d.z.1}.

Recall for a hyperlink $\uL$, we can define a functional $\uPsi_{\uL}^{\ \pm}$ acting on the dual configuration space. As a consequence, we have the following corollary from the main theorem in \cite{EH-Lim05}.

\begin{cor}\label{c.c.1}
Let $S\subset \bR^4$ be an oriented closed connected surface. The spin curvature integrated over the oriented surface $S$ will give us an operator $\hat{F}_S$, which acts on functionals in $\uV^{\ \pm}$.

Suppose $\uPsi_{\oL}^{\ \pm} \in \uV^{\ \pm}$. Then, \beq \hat{F}_S \uPsi_{\uL}^{\ +} =
-i\sqrt{4\pi}\ {\rm lk}(\uL, S)\uPsi_{\uL}^{\ +} \otimes  \mathcal{E}
\nonumber \eeq and \beq \hat{F}_S \uPsi_{\uL}^{\ -} =  -i\sqrt{4\pi}\ {\rm lk}(\uL, S)\uPsi_{\uL}^{\ -} \otimes  \mathcal{E}
. \nonumber \eeq

If $S$ is a compact oriented connected surface with boundary, then the calculations in \cite{EH-Lim05} and using the definition of linking number given in \cite{EH-Lim06}, we see that \beq \hat{F}_S \uPsi_{\uL}^{\ +} =
-\frac{i\sqrt{4\pi}}{4}\ {\rm lk}(\uL, S)\uPsi_{\uL}^{\ +} \otimes  \mathcal{E}
\nonumber \eeq and \beq \hat{F}_S \uPsi_{\uL}^{\ -} =  -\frac{i\sqrt{4\pi}}{4}\ {\rm lk}(\uL, S)\uPsi_{\uL}^{\ -} \otimes  \mathcal{E}
. \nonumber \eeq

\end{cor}

The quantized curvature operator is actually $\mathfrak{su}(2) \times \mathfrak{su}(2)$-valued. Since $\uPsi_{\uL}^{\ \pm}$ is a scalar-valued linear functional on the dual configuration space, we see that $\hat{F}_S \uPsi_{\uL}^{\ \pm}$ is actually $\mathcal{E}$-valued. Note that $\mathcal{E}^\pm$ depends on our earlier choice of $\{\breve{e}_i\}_{i=1}^3$ and $\{\hat{e}_i\}_{i=1}^3$ respectively, so it is not unique. However, its eigenvalues, given by $\pm i\sqrt3/2$ are well-defined, independent of any choice of $\{\breve{e}_i\}_{i=1}^3$ or $\{\hat{e}_i\}_{i=1}^3$.

Thus, the eigenvalues of $\hat{F}_S \uPsi_{\uL}^{\ \pm}$ will be \beq \frac{\sqrt{12\pi}}{2}\ {\rm lk}(\uL, S)\uPsi_{\uL}^{\ \pm}\quad {\rm and}\quad -\frac{\sqrt{12\pi}}{2}\ {\rm lk}(\uL, S)\uPsi_{\uL}^{\ \pm}, \label{e.f.1} \eeq each with algebraic multiplicity 2, when $S$ has no boundary. When $S$ has a boundary, we replace the denominator 2 with an 8 instead.

Each such pair of eigenvalues correspond to the eigenvector $\uPsi_{\uL}^{\ \pm}$, so we will say that the canonical basis $\{\uPsi_{n}^{\ \pm}\}_{n=1}^\infty$ is an eigenbasis for the operator $\hat{F}_S$, with corresponding eigenvalues given by Equation (\ref{e.f.1}). Furthermore, the basis is orthogonal using the counting measure $N$.

Note that the linking number is computed from the oriented time-like triple $\{S, \emptyset, \chi(\oL,\uL)\}$ and the quantum Hilbert space $\{\hat{F}_S, \mathcal{H}(\uV^{\ \pm}) \}$, so it is invariant under the equivalence relation. When $S$ has a boundary, then we need time-like isotopy and time-ordering to ensure that the linking number is proportional to the hyperlinking number between $\partial S$ and $\uL$, so it is invariant under the equivalence relation. In both cases, the linking number is not computed locally in space, unlike area and volume observables.

\begin{rem}\label{r.to.3}
\begin{enumerate}
  \item Unlike the piercing and confinement numbers, the linking number between a geometric hyperlink and a compact surface $S$ with boundary, is actually dependent on any time-ordering between the surface and component geometric loops, when imposed.
  \item When the surface is closed, no time-ordering between the surface and the loop needs to be imposed. But, as explained in Section \ref{s.obs}, a time-ordering is actually implicitly defined when the surface is connected, dependent on the orientation of the loop, which will affect the sign of the linking number accordingly.
\end{enumerate}
\end{rem}


%



Quantized curvature has an important application to cosmology, as explained in \cite{Thiemann:2002nj}. In the Friedmann-Robertson-Walker (FRW) universe model, it is predicted that as time goes down to 0, the universe will be so warped up that curvature is infinity and the metric will be singular. However, in such a regime, volume would be so small that quantum gravity has to take over as the main physics theory. Furthermore, LQG is the correct theory for gravity when metric becomes singular. See Remark \ref{r.tr.4} and also \cite{PhysRevD.77.024046}. The curvature predicted by LQG would always be finite, so the FRW model no longer holds in quantum gravity regime. We would also like to highlight the work done by the authors in \cite{PhysRevLett.96.141301}, who proposed replacing the big bang, with a `quantum bounce', which will also resolve the singularity predicted by the FRW model.

\section{Hamiltonian constraint and Stress operators}\label{s.hmo}

If one looks at Einstein's field equations in GR, one side of the equation is geometric in nature; the other side is not. Except for trivial cases, it is almost impossible to write down the stress-energy tensor. This indeed taints the beauty of the theory of GR. See \cite{feynman2002feynman}.

In \cite{Thiemann:2002nj}, Thiemann also explained that Einstein's equation is flawed in the sense that one side of the equation is classical theory, whereas the other side of the equation is quantum theory. The variable to be solved in Einstein's equation is the metric $g$. He proposed to convert Einstein's equation into an operator-valued equation and solve for a quantized metric. His proposed quantum Einstein's equation can be found in \cite{Thiemann:2007zz}.

In our humble opinion, this is not correct. Rather, we will use Einstein's equation to define a stress operator, using the area and curvature operators defined in \cite{EH-Lim03} and \cite{EH-Lim05} respectively. Also see Section \ref{s.qo}.

We know that Ricci curvature $\Ric$ and scalar curvature ${\rm Scal}$ are defined as \beq \Ric_{\alpha}^\beta = R_{\alpha \gamma}^{\gamma\beta}\quad {\rm and}\quad {\rm Scal} := {\rm Ric}_{\alpha }^\alpha \equiv R_{\alpha \gamma}^{\gamma\alpha}, \nonumber \eeq
whereby $R_{\alpha\beta}^{\gamma\delta}$ are the components of the Riemann curvature tensor.

Einstein's equations are given by ($\Ric_{\mu\gamma} \equiv g_{\mu\beta}\Ric_{\gamma}^{\beta}$) \beq \Ric_{\mu\gamma} - \frac{1}{2}{\rm Scal}\ g_{\mu \gamma} = \kappa T_{\mu \gamma}, \ \kappa = 8\pi G,\label{e.einstein} \eeq whereby $T$ is the stress-energy tensor and $G$ is Newton's gravitational constant. We are going to quantize $T$ into an operator $\hat{T}_{S}^\pm$, which depends on a surface $S$, using the Einstein's equation.

Pick a closed surface $S \subset \{0\} \times \bR^3$. From our quantization of the curvature tensor, we see that we have to integrate $R$ over a given surface. And from Remark \ref{r.c.1}, we see that it is not possible to define a quantized Ricci or Scalar curvature using a surface $S$, as the surface integral would require specifying a metric.

So we replace Ricci curvature with the operator $\hat{F}_{S}$ as defined in Corollary \ref{c.c.1}. But notice that $\hat{F}_{S}$ is actually a multiple of $\mathcal{E} \equiv (\mathcal{E}^+, -\mathcal{E}^-)$.

\begin{notation}
Write \beq \hat{F}_{S}^+ \uPsi_{\uL}^{\ +} = -i\sqrt{4\pi}\ {\rm lk}(\uL, S)\uPsi_{\uL}^{\ +} \otimes  \mathcal{E}^+, \quad \hat{F}_{S}^- \uPsi_{\uL}^{\ -} = i\sqrt{4\pi}\ {\rm lk}(\uL, S)\uPsi_{\uL}^{\ -} \otimes  \mathcal{E}^-. \nonumber \eeq Thus, $\hat{F}_{S} =
\left(
  \begin{array}{cc}
    \hat{F}_{S}^+ &\ 0 \\
    0 &\ \hat{F}_{S}^- \\
  \end{array}
\right)$, which acts on $
\left(
  \begin{array}{c}
    \uPsi_{\uL}^{\ +} \\
    \uPsi_{\uL}^{\ -} \\
  \end{array}
\right)$.
\end{notation}

As for the metric $g$, we replace it with the Area operator $\hat{A}_{S}$ defined in Corollary \ref{c.a.1}. But the Area operator has the real and imaginary eigenvalues, which we can view them as components in a 2-vector.

\begin{notation}
Write \beq
\hat{A}_{S}^+ \oPsi_{\oL}^{\ +} := \frac{|q|\sqrt\pi}{2} \left[
\sum_{u=1}^{\on}\nu_{S}(\ol^u)  \sqrt{\xi_{\rho_u^+}}\right]\oPsi_{\oL}^{\ +}, \quad \hat{A}_S^- \oPsi_{\oL}^{\ -} := \frac{|q|\sqrt\pi}{2} \left[
\sum_{u=1}^{\on}\nu_{S}(\ol^u)\sqrt{\xi_{\rho_u^-}}\right]\oPsi_{\oL}^{\ -}. \nonumber \eeq By abuse of notation, we will now write $\hat{A}_S$ in matrix form, \beq \hat{A}_S =
\left(
  \begin{array}{cc}
    \hat{A}_{S}^+ &\ 0 \\
    0 &\ \hat{A}_{S}^- \\
  \end{array}
\right),\ {\rm which\ acts\ on}\
\left(
  \begin{array}{c}
    \oPsi_{\oL}^{\ +} \\
    \oPsi_{\oL}^{\ -} \\
  \end{array}
\right). \nonumber \eeq
\end{notation}

\begin{defn}(Stress Operator)\label{d.mo}\\
Choose a closed surface $S \subset \{0\} \times \bR^3$. Recall we defined phase spaces \beq V^+ \subset \oV^{\ +} \otimes \uV^{\ +}, \quad V^- \subset \oV^{\ -} \otimes \uV^{\ -}. \nonumber \eeq Define a stress operator $\hat{T}_S$, which acts on the direct product $V^+ \times V^-$,
\begin{align*}
\hat{T}_S &\equiv
\left(
  \begin{array}{cc}
    \hat{T}_{S}^+ &\ 0 \\
    0 &\ \hat{T}_{S}^- \\
  \end{array}
\right)
:= \frac{1}{\kappa}\left( \hat{F}_S - \frac{1}{2}\hat{A}_S\right).
\end{align*}

Suppose $\oL = \{\ol^u\}_{u=1}^{\on}$, each matter loop component colored with representation $\rho_u^+$ and $\rho_u^-$. Then on $V^+ \subset \oV^{\ +} \otimes \uV^{\ +}$, we have
\begin{align*}
\hat{T}_{S}^+\left(\oPsi_{\oL}^{\ +} \otimes \uPsi_{\uL}^{\ +}\right) =& \frac{1}{\kappa}\left(\oPsi_{\oL}^{\ +} \otimes \hat{F}_S^+\uPsi_{\uL}^{\ +} - \frac{1}{2}\hat{A}_{S}^+\oPsi_{\oL}^{\ +}\otimes \uPsi_{\uL}^{\ +}\right),\\
\equiv& \frac{1}{\kappa}\left( -i\sqrt{4\pi}\ {\rm lk}(\uL, S)\mathcal{E}^+ - \frac{|q|\sqrt\pi}{4} \left[
\sum_{u=1}^{\on}\nu_S(\ol^u)  \sqrt{\xi_{\rho_u^+}}\right] \right)\oPsi_{\oL}^{\ +} \otimes \uPsi_{\uL}^{\ +},
\end{align*}
whereby $\kappa = 8\pi G$.

Similarly, on $V^- \subset \oV^{\ -} \otimes \uV^{\ -}$, we have
\begin{align*}
\hat{T}_{S}^-\left(\oPsi_{\oL}^{\ -} \otimes \uPsi_{\uL}^{\ -}\right) =& \frac{1}{\kappa}\left(\oPsi_{\oL}^{\ -} \otimes \hat{F}_S^-\uPsi_{\uL}^{\ -} - \frac{1}{2}\hat{A}_{S}^-\oPsi_{\oL}^{\ -}\otimes \uPsi_{\uL}^{\ -}\right),\\
\equiv& \frac{1}{\kappa}\left( i\sqrt{4\pi}\ {\rm lk}(\uL, S)\mathcal{E}^- - \frac{|q|\sqrt\pi}{4} \left[
\sum_{u=1}^{\on}\nu_S(\ol^u)  \sqrt{\xi_{\rho_u^-}}\right] \right)\oPsi_{\oL}^{\ -} \otimes \uPsi_{\uL}^{\ -}.
\end{align*}

Note that $V^\pm$ is a vector space containing functionals acting on the phase space $\{(\uL_{m}, \oL_m)\}_{m=1}^\infty $. And we are using a time-like triple $\{S, \emptyset, \chi(\oL,\uL)\}$ to compute the linking and piercing numbers of $S$ with $\uL$ and $\oL$ respectively.
\end{defn}

\begin{rem}
We can state a similar definition for a compact connected surface with boundary. Then
\beq \hat{T}_{S}^+\left(\oPsi_{\oL}^{\ +} \otimes \uPsi_{\uL}^{\ +}\right) = \frac{1}{4\kappa}\left( -i\sqrt{4\pi}\ {\rm lk}(\uL, S)\mathcal{E}^+ - |q|\sqrt\pi \left[
\sum_{u=1}^{\on}\nu_S(\ol^u)  \sqrt{\xi_{\rho_u^+}}\right] \right)\oPsi_{\oL}^{\ +} \otimes \uPsi_{\uL}^{\ +}, \nonumber \eeq and \beq \hat{T}_{S}^-\left(\oPsi_{\oL}^{\ -} \otimes \uPsi_{\uL}^{\ -}\right) = \frac{1}{4\kappa}\left( i\sqrt{4\pi}\ {\rm lk}(\uL, S)\mathcal{E}^- - |q|\sqrt\pi \left[
\sum_{u=1}^{\on}\nu_S(\ol^u)  \sqrt{\xi_{\rho_u^-}}\right] \right)\oPsi_{\oL}^{\ -} \otimes \uPsi_{\uL}^{\ -}. \nonumber \eeq
\end{rem}

Note that $\hat{T}_{S}^\pm\left(\oPsi_{\oL}^{\ \pm} \otimes \uPsi_{\uL}^{\ \pm}\right)$ is a matrix-valued linear functional on the phase space. Now $\mathcal{E}^\pm$ has purely imaginary eigenvalues, hence the operator $\hat{T}_{S}^\pm\left(\oPsi_{\oL}^{\ \pm} \otimes \uPsi_{\uL}^{\ \pm}\right)$ has real eigenvalues. For $(\underline{\mathcal{L}},\overline{\mathcal{L}})$ in the phase space, the two real eigenvalues for $\hat{T}_{S}^+\left(\oPsi_{\oL}^{\ +}(\underline{\mathcal{L}}) \otimes \uPsi_{\uL}^{\ +}(\overline{\mathcal{L}})\right)$ are given by \beq
\frac{1}{\kappa}\left(  - \frac{|q|\sqrt\pi}{4} \left[
\sum_{u=1}^{\on}\nu_S(\ol^u)  \sqrt{\xi_{\rho_u^+}}\right]\pm \frac{\sqrt{12\pi}}{2}\ {\rm lk}(\uL, S) \right)\oPsi_{\oL}^{\ +}(\underline{\mathcal{L}}) \otimes \uPsi_{\uL}^{\ +}(\overline{\mathcal{L}}), \nonumber \eeq
and for $\hat{T}_{S}^-\left(\oPsi_{\oL}^{\ -}(\underline{\mathcal{L}}) \otimes \uPsi_{\uL}^{\ -}(\overline{\mathcal{L}})\right)$, given by \beq
\frac{1}{\kappa}\left(  -\frac{|q|\sqrt\pi}{4} \left[
\sum_{u=1}^{\on}\nu_S(\ol^u)  \sqrt{\xi_{\rho_u^-}}\right]\mp \frac{\sqrt{12\pi}}{2}\ {\rm lk}(\uL, S) \right)\oPsi_{\oL}^{\ -}(\underline{\mathcal{L}}) \otimes \uPsi_{\uL}^{\ -}(\overline{\mathcal{L}}). \nonumber \eeq

These eigenvalues depend on how the matter and geometric hyperlinks are linked with the closed surface $S$. Furthermore, we can see quantum fluctuations of the area eigenvalues, due to the geometric hyperlink. We can write a similar expression when $S$ has a boundary.

From the discussion, we will say that the basis vectors $\{ \oPsi_{n}^{\ \pm} \otimes \uPsi_{n}^{\ \pm}\}_{n=1}^\infty $ for $V^\pm$ form an eigenbasis for the operators $\hat{T}_{S}^\pm$. Furthermore, this basis is orthogonal using the counting measure $N$.

The eigenvalues from the area operator were interpreted as total momentum of the matter hyperlink, coming from boost and rotation separately. Curvature of a field is usually interpreted as the field strength and since we integrate the curvature over a surface $S$ and quantize it via a path integral, the eigenvalues of $\hat{F}_S$ should be interpreted as the quantized flux over the surface $S$. See \cite{peskin1995introduction}. Hence, the eigenvalues of the stress operator measure the total quanta of matter and geometric hyperlinks passing through the surface $S$.

\begin{rem}
From Remark \ref{r.to.3}, the linking number is dependent on the imposed time-ordering, when the surface has a boundary. Even when $S$ is connected and has no boundary, the linking number is dependent on the implicitly defined time-ordering, as explained in Section \ref{s.obs}. Hence, the stress operator is actually dependent on time-ordering.
\end{rem}



Observe that the stress operators do not contain any information of energy. So where is the Hamiltonian? Recall we explained that the volume operator gives us the kinetic energy of matter hyperlinks. But in the construction of the stress operator, we see that the two stress operators were defined to be acting on the tensor product space $V^\pm \subset \oV^{\ \pm} \otimes \uV^{\ \pm}$. The volume operator acts only on $\oV^{\ \pm}$, so we must have a second operator that acts on $\uV^{\ \pm}$ and it must describe energy. The missing operator is the potential operator $\hat{U}$, which we will now define.

Modanese in \cite{MODANESE1995697} talked about potential energy in quantum gravity. Borrowing ideas from Quantum Chromodynamics (See \cite{Nair}.), he explained how to compute the energy eigenvalues from a path integral expression. In a background independent environment, time and length do not make sense and the potential energy should refer to geometrical invariant quantities. This motivates the following definition.

\begin{defn}(Potential operator)\label{d.po}\\
Recall in our construction of the phase space, for each matter hyperlink $\oL_m = \{\ol_m^1, \ldots, \ol_m^{\on(m)}\}$, we chose how it is tangled with a corresponding geometric hyperlink $\uL_m$. We denoted this hyperlink as $\chi(\oL_m, \uL_m)$. Now for each component matter loop from $\oL_m$, we will write $\chi(\ol_m^u, \uL_m)$ to denote the hyperlink inside $\chi(\oL_m, \uL_m)$, formed using $\ol_m^u$ and $\uL_m$. Note that $\ol_m^u$ is colored with a representation $\rho_{m,u}$.

We define the potential operator, denoted by $\hat{U}$, as \beq \left[\hat{U}\uPsi_{m}^{\ \pm}\right](\oL_n) :=
\left\{
  \begin{array}{ll}
    \left[\sum_{u=1}^{\overline{n}(m)}\log \ Z^\pm(q; \chi(\{\ol_m^u, \rho_{m,u}\} , \uL_m)) \right]\uPsi_{m}^{\ \pm}(\oL_m), & \hbox{$n=m$;} \\
    0, & \hbox{$n \neq m$.}
  \end{array}
\right. \nonumber \eeq
\end{defn}

\begin{rem}
\begin{enumerate}
  \item So far, all the observables defined in Section \ref{s.qo} were related to a geometric object independent from the hyperlink. For the potential operator however, we are using the geometric hyperlink to define it.
  \item The eigenvalues of $\hat{U}$ are non-negative.
  \item The hyperlinking number is an equivalence invariant, if we impose time-ordering between matter and geometric loops. Hence the eigenvalues of the potential operator are well-defined for this equivalence class. Remark \ref{r.tr.1} says that $Z^\pm$ is actually dependent of this time-ordering, thus the potential operator $\hat{U}$ is actually independent of time-ordering.
\end{enumerate}

\end{rem}

Refer to Section \ref{s.su.1}. From Equation (\ref{e.tr.1}) and Definition \ref{d.z.1}, we leave to the reader to check that indeed, \beq \log \ Z^\pm(q; \chi(\{\ol^u, \rho_{u}\} , \uL)) =
\left\{
  \begin{array}{ll}
    \log\left[\sum_{v=1}^{(2j_{\rho_{u}^\pm}+1)/2}2\cosh\left(\pi q\ {\rm sk}(\ol^u, \uL)\lambda_v^\pm \right) \right], & \hbox{$2j_{\rho_{u}^\pm} + 1$ is even;} \\
    \ \log\left[1 + \sum_{v=1}^{j_{\rho_{u}^\pm}}2\cosh\left(\pi q\ {\rm sk}(\ol^u, \uL)\lambda_v^\pm \right) \right], & \hbox{$2j_{\rho_{u}^\pm} + 1$ is odd.}
  \end{array}
\right. \label{e.l.3} \eeq Here, $\{\pm\lambda_v^\pm \}_{v=1}^{(2j_{\rho^\pm}+1)/2}$ and $\{\pm\lambda_v^\pm \}_{v=1}^{j_{\rho^\pm}}$ are the set of non-zero real eigenvalues of $\rho^\pm(i\mathcal{E}^\pm)$, $\mathcal{E}^\pm$ as defined in Section \ref{s.su.1}, corresponding to $2j_{\rho}^\pm + 1$ being even and odd respectively.

Hence, we see that $\{\uPsi_{m}^{\ \pm}\}_{m=1}^\infty$ forms an eigenbasis for $\hat{U}$, the corresponding eigenvalues given by Equation (\ref{e.l.3}). Note that the eigenvalues are all greater than or equal to 0.

So what does the potential energy operator measure? If one assumes $q$ to be small and we do a Taylor's series expansion, then we have \beq \sum_{u=1}^{\overline{n}}\log \ Z^\pm(q; \chi(\{\ol^u, \rho_{u}\} , \uL)) =
\sum_{u=1}^{\overline{n}}\log(2j_{\rho_{u}^\pm}+1) + q^2\sum_{u=1}^{\overline{n}}\frac{ \pi^2}{(2j_{\rho_{u}^\pm}+1)} \sum_{v \geq 1}\left( {\rm sk}(\ol^u, \uL)\lambda_v^\pm \right)^2 + \cdots, \nonumber \eeq in both cases.

If one compares the above series expansion with the Expression 28 in \cite{MODANESE1995697}, then one sees that the charge $q$ and its coefficient should be interpreted as the mass of a particle and interaction energy respectively. The coefficient is computed using the geometrical invariant hyperlinking number and we can interpret physically as the matter hyperlink $\{l^u\}_{u=1}^{\overline{n}}$ is held together by a geometric hyperlink $\uL$. Thus, gravitational attraction can now be visualized as matter loops being held together by a geometric hyperlink.

We are now ready to define the Hamiltonian of the phase space, which measures pockets of energy in compact regions in space.

\begin{defn}(Hamiltonian constraint operator)\label{d.ham}\\
Choose a compact solid region $R \subset \bR^3$. We define the Hamiltonian constraint operator, $\hat{H}_{R}$ as \beq \hat{H}_{R}  = \hat{V}_R \otimes \underline{1} + \overline{1} \otimes \hat{U}, \nonumber \eeq where $\underline{1}$ and $\overline{1}$ are the identities on $\uV^{\ \pm}$ and $\oV^{\ \pm}$ respectively. This operator clearly acts on $V^\pm$, which is a vector space of functionals acting on the phase space $\{(\uL_{m}, \oL_m)\}_{m=1}^\infty $.
\end{defn}

\begin{rem}\label{r.tr.5}
From Remark \ref{r.to.2}, note that the eigenvalues of the volume operator $\hat{V}_R$ for a solid compact region $R \subset \bR^3$, is independent of any time-ordering.  Because of Remark \ref{r.tr.1}, we see that the operator $\hat{U}$ is dependent on how we time-order the matter and geometric loops.
\end{rem}

Note that $\hat{V}_R$ acts on $\oPsi_{m}^{\ \pm}$, which depends on matter hyperlinks; whereas $\hat{U}$ acts on $\uPsi_{m}^{\ \pm}$, which depends on geometric hyperlinks. Furthermore, we are using the time-like triple $\{\emptyset, R, \chi(\oL, \uL)\}$ to compute the eigenvalues. Hence we have the following nice result.

\begin{thm}
For a compact solid region $R \subset \bR^3$, the basis vectors $\{ \oPsi_{m}^{\ \pm} \otimes \uPsi_{m}^{\ \pm}\}_{m=1}^\infty $ in vector space $V^\pm $ forms an eigenbasis for the Hamiltonian constraint operator $\hat{H}_{R}$. Furthermore, it is an orthogonal basis using the counting measure $N$.
\end{thm}

Thiemann expressed the Hamiltonian constraint as $H(N)$, whereby $N$ is some test function defined on spatial $\bR^3$ and $H(N)$ is a volume integral of $N$ in $\bR^3$, independent of time. See Equation (4.26) in \cite{Thiemann:2006cf}.

By writing the `regularization version' of the said volume integral as a sequence of Riemann sums, Thiemann `promoted' the integrands using rules of canonical quantization, to be operators, hence defining a sequence of operators. To complete the construction, one has to show that it converges to some operator, denoted as $\hat{H}(N)$.

\begin{rem}
Obviously, there is ambiguity as to how one order the terms in the integrands, and this is a problem that one has to face, when using canonical quantization.
\end{rem}

We did not use any test function $N$ to define the Hamiltonian constraint operator. This is because in our construction of observables given in Section \ref{s.obs}, it does not make sense to define a test function, compactly supported in regions $R \subset \bR^3$, which should be thought of as an equivalence class in $\bR^3$, up to spatial diffeomorphism. In \cite{Thiemann:2006cf}, it was also acknowledged that $H(N)$ cannot be spatially diffeomorphism invariant. Thiemann's Hamiltonian constraint operator was also criticized for being `too local' in \cite{Smolin:1996fz}.

In contrast, we have chosen the Hamiltonian constraint operator to be dependent on the choice of the compact solid region $R \subset \bR^3$, up to diffeomorphism equivalence. But the potential operator which we have added into Definition \ref{d.ham}, is actually independent of it. As computation of the hyperlinking number of a hyperlink requires the global topology of the hyperlink, the potential operator cannot be considered a `local' operator. Later in subsection \ref{ss.lpe}, we will explain why the volume operator should be considered as a `local' operator. Thus, the Hamiltonian constraint operator in Definition \ref{d.ham} is a non-local operator.

From Corollary \ref{c.vo.1}, we see that the eigenvalues of $\hat{V}_R$ are always non-negative. We also said that the eigenvalues of $\hat{U}$ are always non-negative. Hence, the Hamiltonian constraint operator defined in Definition \ref{d.ham} will have non-negative eigenvalues. The eigenstate with zero eigenvalue will correspond to the ground state. Geometrically, we can represent the ground state as an empty set, which means there is no matter and geometric hyperlink. If we choose $R = \emptyset$, then for a non-trivial quantum state, we see that the Hamiltonian constraint operator will yield a strictly positive eigenvalue, which is the eigenvalue of the potential operator. Hence we interpret this as there will always be positive energy in the background of space-time.

The potential operator depends on the hyperlinking number between the matter loop and geometric hyperlink. Together with the volume operator, which is defined using a compact region in spatial $\bR^3$, we see that the Hamiltonian constraint operator we have defined earlier is dependent on any time-ordering, as discussed above. This is similar to the quantized curvature operator and the stress operator, both are also dependent on time-ordering. See Remark \ref{r.to.3}.

\begin{rem}\label{r.a3}
Suppose for each $u=1, \ldots, \on$, we have the time-ordering $\ol^u < \uL$ or $\ol^u > \uL$. Thus, we will have \beq
\sum_{v \geq 1}\cosh\left( \pi q\ {\rm sk}(\ol^u, \uL)\lambda_v^\pm \right) = \sum_{v \geq 1}\cosh\left(3\pi q\ {\rm lk}(\pi_0(\ol^u), \pi_0(\uL))\lambda_v^\pm \right), \nonumber \eeq from Remark \ref{r.tr.1}. In such a scenerio, we see that the potential operator and the Hamiltonian constraint operator, will be independent of time-ordering.

In Section \ref{s.fr}, we will discuss how Einstein's equations might allow us to consider such a scenerio.
\end{rem}

\section{Entropy of a black hole}\label{s.bh}

It was shown by Bekenstein and Hawking that the entropy of a black hole should be finite, which was proved within the framework of QFT on curved space-times and should therefore be valid in a semi-classical regime in which quantum fluctuations of the gravitational field are negligible, for example in a large black hole. As remarked by Thiemann in \cite{Thiemann:2002nj}, any successful quantum theory of gravity must obtain this same result.

It was shown in \cite{PhysRevD.9.3292, PhysRevD.7.2333} by Bekenstein that a black hole with area $A$ in space-time should have a Bekenstein entropy \beq S_{BH} = \frac{k}{4}\frac{c^3 A}{G\hbar }, \label{e.ent.1} \eeq 
for a two-dimensional event horizon of the black hole. Here, $G$ is Newton Gravitational's constant, $\hbar$ is Planck's constant, $k$ is Boltzmann's constant and $c$ is speed of light, all set to 1. See also \cite{wald1984general}. Using our construction of the quantum Hilbert space, we are now going to show that entropy is indeed proportional to area.

Let $\Lambda = \frac{\sqrt \pi}{4}\frac{G\hbar}{c^3}$. Krasnov in \cite{krasnov1996statistical} asked, given $A_0 > 0$, how many quantum states are there such that the eigenvalues of the area operator $\hat{A}_H$ lie within the interval $J := [A_0 - \bar{\delta}\Lambda, A_0+ \bar{\delta}\Lambda]$? Here, $A_0 >> \Lambda$ and $\bar{\delta}$ is some small constant, to be determined later. Mathematically, we are asking for the direct sum of the eigenspaces whose eigenvalues lie in the closed interval.

To answer this question, we will make use of the area operator $\hat{A}_S$ defined in \cite{EH-Lim03}, of which we stated its eigenvalues in Corollary \ref{c.a.1}. But note that each eigenstate will give us an eigenvalue $\rho^+$ or $\rho^-$, from eigenstates $\oPsi_{n}^{\ +}$ and $\oPsi_{m}^{\ -}$ respectively. Thus, we will reformulate the above question and ask how many states are there such that \beq \rho^\pm \in  \left[A_0^\pm - \bar{\delta}\Lambda, A_0^\pm + \bar{\delta}\Lambda \right] . \nonumber \eeq We will specify $\bar{\delta}$ shortly.

The eigenvalues are obtained by counting how many times a projected matter hyperlink pierce a surface $S \subset \bR^3$. However, we can maintain how a projected hyperlink pierce the surface $S$, but change its topology far away from the surface. Furthermore, we can add in more trivial knots which does not pierce the surface $S$, while the piercing number remains unchanged. In other words, there will be an infinite number of such states. Therefore, to argue that there should only be a finite number of states corresponding to eigenvalues in a compact interval, we need to make some reasonable assumptions on the event horizon $H$ of a black hole.

In \cite{PhysRevLett.77.3288}, some assumptions were made, of which only 2 will be useful for our discussion.
\begin{description}
  \item[A1]\label{r.a1}\noindent First assumption is that only the configurations of the hole itself, and not the configurations of the surrounding geometry, affects the hole entropy. This means that it suffices to consider only trivial matter hyperlinks, since how a matter hyperlink is tangled will not affect the entropy. It also means that all matter hyperlinks which does not pierce the hole when projected, will and should not be considered as a quantum state.
  \item[A2]\label{r.a2}\noindent Secondly, since we are considering the thermodynamical behaviour of a system containing the hole, we do not have to consider the black hole's interior. This means that if the event horizon is a closed surface, then any link in the interior of the black hole surface, that does not pierce the surface, should not be considered as a quantum state.
\end{description}

Consider an oriented surface $S$, which models the event horizon of a blackhole. For simplicity, we assume that $S$ is connected and compact, possibly with or without boundary. From Assumption A1, the matter hyperlink representing any set of particles emitted out from the event horizon, is the trivial hyperlink.  From Assumption A2, we do not consider any loop strictly inside the interior of a closed surface $S$. Any matter loop in consideration should have a non-trivial representation. For each matter hyperlink, the component matter loop must have either non-trivial translational or angular momentum, i.e. $j_\pm \neq 0$. Finally, we assume that the piercing number between each trivial loop and the event horizon is $\beta>0$. In summary, any loop in a matter hyperlink must have a piercing number $\beta$ with $S$, colored with a non-zero representation.

The states we need to consider in the dual configuration space will be $\{\oL_m\}_{m=1}^\infty$, each $\oL_m$ will be a trivial hyperlink. Orientate each of the hyperlink in a particular manner, which will not affect the eigenvalue of $\hat{A}_S$. Color each matter hyperlink $\oL_m$ with a set of representations $\{\rho_{m,u} \equiv (\rho_{m,u}^+, \rho_{m,u}^-)\}_{u\geq 1}$. Construct a Quantum Loop Space using $\{\oL_m\}_{m=1}^\infty$, denoted by $\mathcal{H}(\oV^{\ \pm})$, as defined in Definition \ref{d.qls}. Each quantum state $\oL_m$ is represented by a set of unlinked trivial loops, each loop colored by a representation, indexed by $(j^+, j^-)$, with $j^\pm > 0$.

Using Assumptions A1 and A2, each quantum state $\oL_m$ in $\mathcal{H}(\oV^{\ \pm})$ is then represented by a $n$-tuple $\vec{p}^\pm = (p_1^\pm, \ldots, p_n^\pm)$, with arbitrary $n$, denoting the number of trivial loops in $\oL_m$. Each $0<p_u^\pm := 2j_u^\pm$, whereby $j_u^\pm$ is the half-integer representation of the loop $\ol^u$. States labelled with different orderings of the same unordered $n$-tuples of integers $\vec{p}^\pm$ are distinguishable for an external observable.

Now the eigenvalues for the area operator we obtained in Corollary \ref{c.a.1} are either purely real or purely imaginary. We will now consider the real case. The imaginary case is similar. We will also drop the superscript `$+$' and ask, how many states in $\mathcal{H}(\oV^{\ +})$ have real eigenvalues in the interval $\left[A_0 - \bar{\delta}\Lambda, A_0 + \bar{\delta}\Lambda \right]$?

By putting in all the physical constants, our eigenvalue of the area operator for a particular quantum state $\{\oL_m, \rho_m\} \in \mathcal{H}(\oV^{\ +})$ is \beq |q|\beta\frac{\sqrt \pi}{2}\frac{\hbar G}{c^3}\sum_{u \geq 1}\sqrt{j_u(j_u+1)} \equiv |q|\beta\frac{\sqrt \pi}{4}\frac{\hbar G}{c^3}\sum_{u \geq 1}\sqrt{p_u(p_u+2)},\ p_u=2j_u. \nonumber \eeq Now $j_u > 0$ is a half-integer, so $p_u$ is a positive integer.

In \cite{Krasnov:1996wc}, the area operator is actually scaled by some undetermined constant $\gamma$. One can find a similar expression in \cite{Ashtekar:2000eq}. Here, our undetermined constant is actually the charge $q$. In Section \ref{s.hmo}, the charge $q$ has a nice interpretation as mass, which is intrinsic to the matter hyperlink and has nothing to do with the quantum states. Thus, we will drop $q$ from the area operator in the rest of this section.

The proof we are presenting here, is taken from \cite{PhysRevLett.77.3288}, but with some modifications. Let $\delta := 0.55$. Note that $\bigcup_{M \in \mathbb{N}}[M-\delta, M + \delta] \supset [1, \infty)$. For each natural number $M >0$, define the following sets,
\begin{align*}
A(M) :=& \left\{ \vec{p} = (p_1, \ldots, p_n):\ p_u \in \mathbb{N}\ {\rm and}\ M-\delta\leq \sum_{u=1}^n \sqrt{p_u(p_u+2)} \leq M + \delta\right\}, \\
A^+(M) :=& \left\{ \vec{p} = (p_1, \ldots, p_n):\ p_u \in \mathbb{N}\ {\rm and}\ \sum_{u=1}^n p_u = M \right\}, \\
A^-(M) :=& \left\{ \vec{p} = (p_1, \ldots, p_n):\ p_u \in \mathbb{N}\ {\rm and}\ \sum_{u=1}^n (p_u+1) = M \right\}.
\end{align*}

Let $N(M)$ be the number of ordered $n$-tuples in $A(M)$; let $N_\pm(M)$ be the number of ordered $n$-tuples in $A^\pm(M)$.

\begin{lem}\label{l.bh.2}
Let $p_+ = p + 1$, $p$ is a natural number. Then we have
\beq 1 \leq \sqrt{(p_+ + 1)^2-1} - \sqrt{(p+1)^2-1} \leq \sqrt8 - \sqrt{3} < 1.1= 2\delta. \label{e.lhc.1} \eeq
\end{lem}

\begin{proof}
Note that $p_u(p_u+2) = (p_u + 1)^2-1$. Suppose $(p_1, \cdots, p_n) \in A(m)$ for some $m \leq M$. Define \beq f(p) = \sqrt{(p_+ + 1)^2-1} - \sqrt{(p+1)^2-1} > 0,\ p \geq 1. \nonumber \eeq

For the lower bound, note that
\begin{align*}
\sqrt{(p_+ + 1)^2 - 1} - \sqrt{(p+1)^2 - 1} =& \frac{(p_+ + 1)^2 - 1 - (p+1)^2 + 1}{\sqrt{(p_+ + 1)^2 - 1} + \sqrt{(p+1)^2 - 1}} \\
\geq& \frac{p^2+4p + 4 - p^2 - 2p - 1}{p_+ + 1 + p + 1} \\
=& \frac{2p+3}{2p + 3} \geq 1.
\end{align*}

Take the derivative,
\begin{align*}
f'(p) =& \frac{p_+ + 1}{\sqrt{(p_+ + 1)^2 - 1}} - \frac{p+1}{\sqrt{(p+1)^2 - 1}} \\
=& \frac{(p_+ + 1)\sqrt{(p+1)^2 - 1} - (p+1)\sqrt{(p_+ + 1)^2 - 1}}{\sqrt{(p_+ + 1)^2 - 1}\sqrt{(p+1)^2 - 1}}.
\end{align*}

Using the lower bound $f(p) \geq 1$, we have
\begin{align*}
(p_+ + 1)&\sqrt{(p+1)^2 - 1} - (p+1)\sqrt{(p_+ + 1)^2 - 1} \\
=& (p+1)\left[\sqrt{(p+1)^2 - 1} -  \sqrt{(p_+ + 1)^2 - 1}\right] + \sqrt{(p+1)^2 - 1} \\
\leq& -(p+1) + \sqrt{(p+1)^2 - 1} < 0.
\end{align*}

Hence, $f'(p) < 0$ and for $p \in \mathbb{N}$, we obtain an upper bound for $f$,
\beq \sqrt{(p_+ + 1)^2-1} - \sqrt{(p+1)^2-1} \leq \sqrt8 - \sqrt{3} < 1.1= 2\delta. \nonumber \eeq

Put the lower and upper bounds together, we thus have Equation (\ref{e.lhc.1}).
\end{proof}

\begin{lem}\label{l.bh.1}
For $M$ large enough, there exists a $d \equiv d(M)$, $0.48 < d < 0.69$ such that we have $\log N(M) = dM$.
\end{lem}

\begin{proof}
First, we find an upper bound $N(M)$. Define $F: A(M) \rightarrow A^+(M)$ as follows. Consider any $n$-tuple $(p_1, \ldots, p_n) \in A(M)$. Since \beq \sum_{u=1}^n \sqrt{p_u(p_u+2)} \geq \sum_{u=1}^n p_u ,\nonumber \eeq thus $\sum_{u\geq 1}p_u = m$ for some unique $m \leq M$. Then \beq
F: (p_1, \ldots, p_n) \longmapsto (p_1, \ldots, p_n+K), \nonumber \eeq with $K$ being the smallest whole number such that $(p_1, \ldots, p_n + K) \in A^+(M)$. Such a $K$ exists because $m \leq M$. Clearly, $F$ is well-defined.

Suppose $F(p_1, \ldots, p_n) = F(q_1, \ldots, q_n)$. By definition of $F$, we must have $p_u = q_u$ for $1\leq u \leq n-1$ and there exist $K$ and $\tilde{K}$ such that $p_n + K = q_n + \tilde{K}$. If $K \neq \tilde{K}$ and hence $p_n \neq q_n$, then the above Equation (\ref{e.lhc.1}) says that $|p_n - q_n| \leq 1$. This means that at most 2 elements in $A(M)$ can be mapped to the same element in $A^+(M)$.

Hence $N(M) \leq 2N_+(M)$. From the proof in \cite{PhysRevLett.77.3288}, we see that $N_+(M) = (C/2)2^M$ for some constant $C$. Therefore, we have $N(M) \leq 2N_+(M) = C2^{M}$.

Now we find a lower bound for $N(M)$. Define $G: A^-(M) \rightarrow A(M)$ as follows. Consider any $n$-tuple $(p_1, \ldots, p_n) \in A^-(M)$. Since \beq \sum_{u=1}^n \sqrt{p_u(p_u+2)} = \sum_{u=1}^n \sqrt{(p_u+1)^2-1} \leq \sum_{u=1}^n (p_u+1)= M, \nonumber \eeq we have $(p_1, \ldots, p_n) \in A(m)$ for some $m \leq M$.

Then \beq G: (p_1, \ldots, p_n) \longmapsto (p_1, \ldots, p_n+L), \nonumber \eeq with $L$ being the smallest whole number such that $(p_1, \ldots, p_n + L) \in A(M)$. Such a $L$ exists because $m \leq M$ and it  is well-defined.

Suppose $(p_1, \ldots, p_n), (q_1, \ldots, q_n) \in A^-(M)$ such that $G(p_1, \ldots, p_n) = G(q_1, \ldots, q_n)$. By definition of $G$, we must have $p_u = q_u$ for $1\leq u \leq n-1$. Since $\sum_{u=1}^n (p_u+1)= \sum_{u=1}^n (q_u+1) = M$, we see that $p_n = q_n$, hence $G$ is an injective map.

Therefore, $N_-(M) \leq N(M)$. From the proof in \cite{PhysRevLett.77.3288}, we see that for some constants $D$ and $E$, \beq N_-(M) = Da_+^M + Ea_-^M, \nonumber \eeq whereby $a_\pm$ are the roots to the equation $x^2 = x + 1$. We let $a_+ = (1+\sqrt5)/2$ be the larger root.

When $M$ is large, $a_+^M$ will dominate, so we will have the lower bound $N(M) \geq Da_+^{M}$. Putting all the bounds together, we will have \beq Da_+^{M} \leq N(M) \leq C2^{M}, \nonumber \eeq for large enough $M$. Take $\log$, we will have \beq \log D  + M\log a_+ \leq \log N(M) \leq \log C + M\log 2 . \nonumber \eeq

Thus, for large enough $M$, there exists a $\log[(1+\sqrt5)/2] < d(M) < \log 2$ such that $\log N(M) = d(M)M$.
\end{proof}

Given $M \in \mathbb{N}$, we note that \beq [A_0 - \bar{\delta}\Lambda , A_0 + \bar{\delta}\Lambda] = \beta\Lambda[M-\delta, M+\delta], \nonumber \eeq for some \beq A_0 = \beta\Lambda M , \quad \bar{\delta} := \beta\delta. \nonumber \eeq

Let the number of quantum states whose real eigenvalues lie inside $[A_0 - \bar{\delta}\Lambda, A_0 + \bar{\delta}\Lambda]$ be denoted by $\tilde{N}(A_0)$. From Lemma \ref{l.bh.1}, we will have \beq \log \tilde{N}(A_0) = \log N(M) = \frac{4d(M)}{\beta\sqrt \pi} \frac{c^3}{\hbar G} A_0. \nonumber \eeq Thus entropy \beq S= k\log \tilde{N}(A_0)= \left( \frac{4d(M)}{\beta\sqrt \pi} \right) \frac{kc^3}{\hbar G}A_0. \nonumber \eeq

\begin{rem}
We refer the reader to Equation (47) in \cite{Rovelli1998}, whereby the entropy is given by \beq S = \frac{\mathbf{c}}{\gamma}\frac{k}{4\hbar G}A, \nonumber \eeq and $\mathbf{c} \sim 1/4\pi$. There is a constant $\gamma$, for which the meaning is unclear. If one compares this expression with our expression, then one may try to identify $\gamma$ with $\beta$. But in our expression, $\beta$ is the piercing number, which has no association with $\gamma$.
\end{rem}

%

A simple calculation will give us \beq 1.08 < \frac{4d(M)}{\sqrt \pi} < 1.56. \nonumber \eeq Now, $1.56/6 \approx 0.26$ and $1.08/5 \approx 0.216$. Thus by choosing an appropriate piercing number $\beta$, the proportionality constant is somewhere between \beq 0.216 < \frac{4d}{\beta\sqrt \pi} < 0.26. \nonumber \eeq

So, we showed that entropy of a black hole is proportional to its area. The constant obtained by Bekenstein was 0.25, but the constant we obtained lies between 0.216 and 0.26. The authors in \cite{Ashtekar:1997yu} gave a partial explanation for this constant, but remarked that the full significance of this proportionality constant is still not well-understood. Alternatively, we could have made some choice of $q$ to obtain the exact proportionality constant, as done in \cite{Ashtekar:2000eq}.

We would like to mention the work done by the authors in \cite{Ashtekar:2000eq}, whereby they obtained the correct constant 0.25 in the Bekenstein-Hawking formula. They did this by introducing the Barbero-Immirzi $\gamma$ to the area operator and defined it accordingly to obtain the right constant. Note that in obtaining a lower bound for the counting of states, they only consider states whereby each $j_u = 1/2$.

\section{Final Remarks}\label{s.fr}

We would like to conclude this article with the following remarks.

\subsection{Gravitons}\label{ss.gv}

Gravitons are predicted to exist in QFT. However, using spin networks, gravitons play little or no role. See \cite{Thiemann:2002nj}. But this is not the case in loop representation of quantum gravity.

We have already explained that the matter hyperlink should represent fundamental particles. And we hope that we have given enough reasons to support this argument. But till now, we have not mentioned anything about the geometric hyperlink. So what does it represent?

In GR, gravity is not a force; but rather a manifestation of how space-time is curved. See \cite{schutz1985first}. To calculate this curvature, one has to obtain the metric $g$, a dynamical variable that is to be solved from Einstein's equations. From Equation (\ref{e.l.7}), we see that the geometric hyperlink is responsible for carrying information of the metric $g$. We postulate that the projection of the geometric hyperlink, $\pi_0(\uL)$, should represent the orbits of another particle, which we will term it as graviton. Now, the graviton at the time of writing, has not been discovered experimentally. In theory, it should have zero mass, just like the photon. See \cite{feynman2002feynman}.

In \cite{CS-Lim01}, we computed the Chern-Simons path integral for the Abelian gauge group $U(1)$. Now, it is known that the Abelian group in Yang-Mills Theory describes the Electromagnetic theory. The photons are the force carrying particles responsible for the electromagnetic force.

In Section \ref{s.cvs}, we explained that even though the vierbein $e$ is not a gauge, however, it can be viewed as a connection using an Abelian gauge group, the group of translations in a 4-dimensional vector space $V$. If we compare with the $U(1)$ gauge theory, then we see that the graviton should also have zero mass.


The Wilson Loop observable of a colored hyperlink $\chi(\oL, \uL)$ is only dependent on the hyperlinking number between each component matter loop $\ol^u$ and the geometric hyperlink $\uL$. See Definition \ref{d.z.1}. From our construction in Notation \ref{n.l.1}, we see that gravitons are produced when there are particles in space. If there are no particles, which imply the absence of matter hyperlinks, then any geometric hyperlink will just yield 1 for the Wilson Loop observable. Loop representation of quantum gravity collapses to a trivial theory.

So it seems that by having particles, we create gravitons, which carry information about the metric. Is it possible for gravitons to create particles? If one looks at Einstein's Equation (\ref{e.einstein}), we see that the particles define the stress-energy tensor which appears on the RHS of the equation. Therefore, his equations forbid the gravitons from creating the particles. Furthermore, the matter hyperlink carries information on the representation of each loop; geometric hyperlinks are not colored. Hence, we do not think it is possible for gravitons to create particles.

Given a set of particles, how do we obtain the colored matter hyperlink $\oL$ that describes these particles and a geometric hyperlink $\uL$ that carries information about the metric and is tangled with $\oL$ to form $\chi(\oL, \uL)$? We do not have an answer to this and probably the answer can only be found in QFT or in some other theory, which is not within the framework of our current work on quantum gravity. But this means that in future, it may be necessary to redefine our construction on loop representation of quantum gravity.

The relationship between matter hyperlink and geometric hyperlink was also implied in the stress operator defined in Definition \ref{d.mo}. This operator was defined on the tensor product $\oV^{\ \pm} \otimes \uV^{\ \pm}$. This means that both hyperlinks must be considered together. Furthermore, the Hamiltonian constraint operator defined in Definition \ref{d.ham}, was also defined on this tensor product.

At the time of writing, gravitational waves were discovered experimentally. Gravitons have yet to be discovered. So how does one attempt to find them? One possible way would be via the experimental values of the curvature operator $\hat{F}_S$ of a surface $S$, and the potential operator $\hat{U}$. To measure discrete values for the curvature of a surface $S$ and potential energy in a small region $R$ in space, the surface and region  must have dimensions of the order of Planck's length. There might be experimental difficulties trying to measure curvature and potential as described above. To see the effects of quantum gravity, we may have to go down to Planck's length or near a Big Bang singularity. And this could be the reason why till now we are unable to verify the existence of gravitons experimentally.

\subsection{Causality}

We refer the reader to \cite{THOOFT2001157}, whereby the author explained how causality is violated when one attempts to quantize gravity. Because there is no preferred metric, it is not possible to have a time scale in LQG. But we saw throughout this article, we have a concept of time-ordering in place of it.

The time-ordering between the submanifolds discussed throughout this article, imply causality, which we will discuss in detail here. Under the time-like triple equivalence relation, causality will never be violated. This means we only consider homeomorphisms of $\bR \times \bR^3$, that respects causality. This is necessary in special relativity, as causality is also a consequence of the fact that information cannot travel faster than the speed of light.

For simplicity, consider a matter loop $\ol$ and a geometric loop $\ul$, together form a time-like hyperlink $\chi(\ol, \ul)$. If $\ol < \ul$, then $\ol$ will be interpreted as the cause and $\ul$ is the effect. If $\ol$ represents a particle and $\ul$ represents a graviton, then we may say that a particle moving in an orbit, creates a graviton. This creation of graviton is not instantaneous, as there must be a time-lag between the two loops.

But the construction in LQG do not forbid $\ol > \ul$. So what does it mean by $\ul$ is the cause and $\ol$ is the effect? Based on the previous paragraph, one may boldly say that a graviton may create matter in space. But we discussed in the previous section that Einstein's equations may forbid gravitons from creating matter.

Suppose matter loops always occur before, or always occur after geometric loops. For a given $u=1,\cdots, \overline{n}$, if $\ol^u < \uL$ or $\ol^u > \uL$, then we see that ${\rm sk}(\ol^u, \uL) = \pm 3\times {\rm lk}(\pi_0(\ol^u), \pi_0(\uL))$ respectively, and hence
\begin{align*}
Z^+(q;  \chi(\oL, \uL) ) =& \prod_{u=1}^{\on}\ \Tr_{\rho^{+}_u}\ \exp\left[\mp 3\pi iq\ {\rm lk}(\pi_0(\ol^u), \pi_0(\uL)) \cdot \mathcal{E}^{+}\right], \\
Z^-(q;  \chi(\oL, \uL) ) =& \prod_{u=1}^{\on}\ \Tr_{\rho^{-}_u}\ \exp\left[\mp 3\pi iq\ {\rm lk}(\pi_0(\ol^u), \pi_0(\uL)) \cdot \mathcal{E}^{-}\right].
\end{align*}
Because $i\mathcal{E}^\pm$ has both positive and negative eigenvalues of the same magnitude, we see that $Z^\pm(q;  \chi(\oL, \uL) )$ will be independent of the time-ordering in this case. Refer to Equation (\ref{e.tr.1}).

As a consequence, we see that the potential operator $\hat{U}$ and hence the Hamiltonian constraint operator defined in Section \ref{s.hmo}, will now be independent of time-ordering. Refer to Remarks \ref{r.tr.5} and \ref{r.a3}. LQG allows the matter hyperlink to occur before or after geometric hyperlink. But if $\oL$ is always the cause and $\uL$ is always the effect, as implied by Einstein's equations, then we see that the imposed time-ordering between matter and geometric loops, is no longer unbiased.

Now let us discuss the case of a compact solid region $R$. In the measurement of the volume of $R$, one is effectively measuring the kinetic energy of particles passing through the interior of $R$. In Remark \ref{r.tr.2}, we remarked that there is a time-ordering implicitly defined on the nodes, i.e. the nodes appear either before or after time $x_0 = 0$.

Suppose that the nodes in the interior of $R$ appear before time 0. So, the particles had already travelled through spatial region $R$. At time 0, we attempt to measure the volume of $R$, the region whereby the particles supposedly pass through. The experimental values we obtain will be the kinetic energies of these particles. So, the cause would be the trajectories of these particles; the effect would be the kinetic energies obtained, which would also give us the volume of $R$. Note that there is a time-lag between the two events. It is not instantaneous, as information needs time to travel. In the case when the nodes appear after time 0, we do not have a reasonable interpretation of cause and effect.

The situation for a surface $S$ with boundary is similar. When $\ol < S$ and $\ul < S$, then we may interpret the cause as the matter particles and gravitons, all pass through the surface before the surface is formed; the effect would be the experimental values of their momentum, which in turn give us quantized area and quantized curvature values of $S$ respectively. When $\ol > S$ or $\ul > S$, then loops will still have an effect on the area or curvature of the surface $S$, as permitted in LQG.


When $S$ is closed surface, it is a different situation altogether. Suppose $S$ and $\pi_0(S)$ are both closed and connected. Recall in Section \ref{s.obs}, we said an implicit time-ordering is implied, which depends on whether the oriented projected arc joining a left or right piercing, lies in the interior or exterior of connected $\pi_0(S)$. In the former (latter), we interpret that a particle enters (exits) the interior of $\pi_0(S)$ before the formation of the closed surface, and it exits (enters) the interior of the closed surface after the formation of the surface. The casual event consists of the connected closed surface $\pi_0(S)$ being formed during the time period between the entry and exit points of the closed surface; the effect would be the surface $S$ having non-trivial curvature and area eigenvalues, computed from the linking number and piercing number respectively.

Both scenerios are equally possible to happen under the framework of LQG. But because $\pi_0(S)$ divides space $\bR^3$ into interior and exterior of $\pi_0(S)$, we see that GR as a local theory, will dictate that the eigenvalues are computed or measured, when the particle enters the interior of the closed surface during the time period of the formation of the surface. Hence, it breaks the symmetry in the flow of time.

In all the cases as discussed above, a time-lag between events exist, and the cause and effect as defined, is unchanged under the equivalence relations we discussed in Section \ref{s.obs}. The mathematics of LQG do not favor a preferred ordering, hence it does not imply an arrow of time. However, GR seems to imply a particular ordering is actually preferred, under certain circumstances.

\begin{rem}
As remarked in \cite{Rovelli1998}, Penrose argued that time should be asymmetric.
\end{rem}

\subsection{Locality and Principle of Equivalence}\label{ss.lpe}

The Wilson Loop observable $Z(q; \chi(\oL, \uL)$ of a colored hyperlink $\chi(\oL, \uL)$ depends on the global topology of the hyperlink $\chi(\oL, \uL)$. It cannot be determined by looking at the hyperlink locally. The same applies to the linking number of a link in $\bR^3$, and the linking number between a compact surface (with or without boundary) and a loop in $\bR^4$. This means physically it is not possible to determine the quantum state by local observation. This is consistent with the fact that quantum theory is a global theory. See \cite{Okon1}.

But GR is a local theory (See \cite{Okon1}.), so does LGQ contradict this? The answer is no. In LQG, recall that axial gauge fixing makes the metric degenerate, hence there is no metric defined on $\bR \times \bR^3$. Therefore, it makes no sense to talk about locality in LQG, as argued in \cite{THOOFT2001157}.

We have a different view on this. We discussed earlier (See subsections \ref{ss.ao} and \ref{ss.vo}.) that a surface can contain a quanta of area and a compact solid region can contain a quanta of volume. We can break up a surface into smaller pieces, each containing a zero or a quanta of area. A solid region can be broken down into smaller blocks, each is a zero or a quanta of volume. The area and volume observables will actually show that LQG is a local theory. In this sense, we will say that the area and volume operators are `local operators'.

On the other hand, quantized curvature, which depends on the linking number between a hyperlink and a surface, is not an observable that can be computed locally in space, as we had discussed earlier. See subsection \ref{ss.co}. Similarly, the hyperlinking number depends on the global topology of the hyperlink. Thus it shows that LQG is a global theory. Therefore, we can say that the quantized curvature operator and the potential energy operator, are `global operators'. The quantized stress operator and Hamiltonian constraint operator, will then be `global operators'. Indeed, LQG is both a local and global theory.

The theory of GR is built upon the Principle of Equivalence, which states that when gravitational effects are observed, it is then impossible, by any experiment whatsoever, to determine the type of gravitational field responsible for it. See \cite{feynman2002feynman}.

So, can we formulate an equivalence principle in LQG? The answer is a yes, and let us re-visit the quantized curvature operator. By observing the surface, we obtain the quantized curvature, which is proportional to the linking number of the surface with (geometric) hyperlinks. If the linking number is non-zero, one can conclude that gravitons are in action and therefore, a gravitational field is present. If the equivalence principle is to hold, it means that the eigenvalues of the quantized curvature are not sufficient to determine the current quantum state responsible for the effects of gravity.

We can formulate the equivalence principle in LQG in the following way.

\begin{defn}
Let $\{S_1, \cdots, S_n\}$ be compact surfaces and suppose $\{ \lambda_1, \cdots, \lambda_n \}$ is any set of observed eigenvalues corresponding to the quantized curvature operators $\hat{F}_{S_1}, \cdots, \hat{F}_{S_n}$ respectively. Suppose $E(\lambda_i) \subset \mathcal{H}(\uV^{\ \pm})$ is the eigenspace of $\hat{F}_{S_i}$, corresponding to this number $\lambda_i$.

We say that the equivalence principle holds on the quantum Hilbert space \beq \left\{ \hat{F}_{S_1}, \cdots, \hat{F}_{S_n}, \mathcal{H}(\uV^{\ \pm}) \right\}, \nonumber \eeq if we have that the dimension of $\bigcap_{i=1}^n E(\lambda_i)$ is at least 2 or more.
\end{defn}

\begin{rem}
This means that even though we obtain quantized curvature eigenvalues $\lambda_1, \cdots, \lambda_n$ from the surfaces, we are still not able to determine uniquely (up to a constant), the quantum eigenstate $\uPsi_{\uL}^{\ \pm} \in \uV^{\ \pm}$, responsible for the set of observed eigenvalues.
\end{rem}

When one obtains a non-zero value from the area or volume operator, then it is enough to conclude that matter hyperlinks are present. Furthermore, the authors in \cite{Rovelli:1995ac} claimed that volume and area operators are enough to distinguish all spin networks from each other. We can formulate this statement mathematically in the following way.

Let $S_1, \cdots, S_n$ be compact surfaces, $\{\lambda_1, \cdots, \lambda_n\}$ be any set of observed eigenvalues corresponding to the area operators $\hat{A}_{S_1}, \cdots, \hat{A}_{S_n}$ respectively, with the corresponding eigenspaces denoted by $E_1, \cdots, E_n \subset \mathcal{H}(\oV^{\ \pm})$.

Let $R_1, \cdots, R_m$ be compact solid regions, and $\{\mu_1, \cdots, \mu_m\}$ be any set of observed eigenvalues corresponding to the volume operators $\hat{V}_{R_1}, \cdots, \hat{V}_{R_m}$ respectively, the corresponding eigenspaces denoted by $F_1, \cdots, F_m \subset \mathcal{H}(\oV^{\ \pm})$.

Suppose on the Hilbert space \beq \left\{\hat{A}_{S_1}, \cdots, \hat{A}_{S_n}, \hat{V}_{R_1}, \cdots, \hat{V}_{R_m}, \mathcal{H}(\oV^{\ \pm}) \right\}, \nonumber \eeq
$\left[\bigcap_{i=1}^n E_i\right] \cap \left[\bigcap_{i=1}^m F_i\right]$ has dimension one. This means that the eigenvalues from volume and area operators of finitely many compact surfaces and compact solid regions, will be sufficient to determine uniquely (up to a constant) the quantum eigenstate $\oPsi_{\oL}^{\ \pm} \in \mathcal{H}(\oV^{\ \pm})$, which gives us the observed eigenvalues for area and volume.


\begin{thebibliography}{10}

\bibitem{2012arXiv1203.3991S}
S.~J. {Summers}, ``{A Perspective on Constructive Quantum Field Theory},'' {\em
  ArXiv e-prints}, Mar. 2012.

\bibitem{Jaffe}
A.~Jaffe, ``Constructive quantum field theory,'' {\em Mathematical Physics
  2000}, pp.~111--127, 2000.

\bibitem{Thiemann:2002nj}
T.~Thiemann, ``{Lectures on loop quantum gravity},'' {\em Lect. Notes Phys.},
  vol.~631, pp.~41--135, 2003.
\newblock [41(2002)].

\bibitem{Thiemann:2007zz}
T.~Thiemann, {\em {Modern canonical quantum general relativity}}.
\newblock Cambridge University Press, 2008.

\bibitem{rovelli2004quantum}
C.~Rovelli, {\em Quantum Gravity}.
\newblock Cambridge Monographs on Mathematical Physics, Cambridge University
  Press, 2004.

\bibitem{Rovelli1998}
C.~Rovelli, ``Loop quantum gravity,'' {\em Living Reviews in Relativity},
  vol.~1, no.~1, p.~75, 1998.

\bibitem{Smolin2006-SMOTCF}
L.~Smolin, ``The case for background independence,'' in {\em The Structural
  Foundations of Quantum Gravity} (D.~Rickles, S.~French, and J.~Saatsi, eds.),
  pp.~196--239, Oxford University Press, 2006.

\bibitem{Ashtekar:2000eq}
A.~Ashtekar, J.~C. Baez, and K.~Krasnov, ``{Quantum geometry of isolated
  horizons and black hole entropy},'' {\em Adv. Theor. Math. Phys.}, vol.~4,
  pp.~1--94, 2000.

\bibitem{THOOFT2001157}
G.~'t~Hooft, ``Obstacles on the way towards the quantisation of space, time and
  matter — and possible resolutions,'' {\em Studies in History and Philosophy
  of Science Part B: Studies in History and Philosophy of Modern Physics},
  vol.~32, no.~2, pp.~157 -- 180, 2001.
\newblock Spacetime, Fields and Understanding: Persepectives on Quantum Field.

\bibitem{MR990772}
E.~Witten, ``Quantum field theory and the {J}ones polynomial,'' {\em Comm.
  Math. Phys.}, vol.~121, no.~3, pp.~351--399, 1989.

\bibitem{CS-Lim02}
A.~P.~C. Lim, ``Non-abelian gauge theory for {C}hern-{S}imons path integral on
  ${R}^3$,'' {\em Journal of Knot Theory and its Ramifications}, vol.~21,
  no.~4, 2012.

\bibitem{Witten:1988hc}
E.~Witten, ``{(2+1)-Dimensional Gravity as an Exactly Soluble System},'' {\em
  Nucl. Phys.}, vol.~B311, p.~46, 1988.

\bibitem{PMIHES_1988__68__175_0}
M.~F. Atiyah, ``Topological quantum field theory,'' {\em Publications
  Math\'ematiques de l'IH\'ES}, vol.~68, pp.~175--186, 1988.

\bibitem{atiyah_1990}
M.~Atiyah, {\em The Geometry and Physics of Knots}.
\newblock Lezioni Lincee, Cambridge University Press, 1990.

\bibitem{PhysRevLett.61.1155}
C.~Rovelli and L.~Smolin, ``Knot theory and quantum gravity,'' {\em Phys. Rev.
  Lett.}, vol.~61, pp.~1155--1158, Sep 1988.

\bibitem{Mercuri:2010xz}
S.~Mercuri, ``{Introduction to Loop Quantum Gravity},'' {\em PoS}, vol.~ISFTG,
  p.~016, 2009.

\bibitem{Ashtekar:2004vs}
A.~Ashtekar, ``{Gravity and the quantum},'' {\em New J. Phys.}, vol.~7, p.~198,
  2005.

\bibitem{1990NuPhB.331...80R}
C.~{Rovelli} and L.~{Smolin}, ``{Loop space representation of quantum general
  relativity},'' {\em Nuclear Physics B}, vol.~331, pp.~80--152, Feb. 1990.

\bibitem{Carlip:1995zj}
S.~Carlip, ``{Lectures on (2+1) dimensional gravity},'' {\em J. Korean Phys.
  Soc.}, vol.~28, pp.~S447--S467, 1995.

\bibitem{Baez:1999sr}
J.~C. Baez, ``{An Introduction to spin foam models of quantum gravity and BF
  theory},'' {\em Lect. Notes Phys.}, vol.~543, pp.~25--94, 2000.

\bibitem{Baez:1999:Online}
J.~Baez, {\em Spin Networks, Spin Foams and Quantum Gravity}, May 1999.

\bibitem{Baez1996253}
J.~C. Baez, ``Spin networks in gauge theory,'' {\em Advances in Mathematics},
  vol.~117, no.~2, pp.~253 -- 272, 1996.

\bibitem{Rovelli:1995ac}
C.~Rovelli and L.~Smolin, ``{Spin networks and quantum gravity},'' {\em Phys.
  Rev.}, vol.~D52, pp.~5743--5759, 1995.

\bibitem{doi:10.1063/1.1587095}
J.~Pullin, ``Canonical quantization of general relativity: the last 18 years in
  a nutshell,'' {\em AIP Conference Proceedings}, vol.~668, no.~1,
  pp.~141--153, 2003.

\bibitem{Thiemann:2006cf}
T.~Thiemann, ``{Loop Quantum Gravity: An Inside View},'' {\em Lect. Notes
  Phys.}, vol.~721, pp.~185--263, 2007.

\bibitem{PhysRevLett.57.2244}
A.~Ashtekar, ``New variables for classical and quantum gravity,'' {\em Phys.
  Rev. Lett.}, vol.~57, pp.~2244--2247, Nov 1986.

\bibitem{PhysRevD.36.1587}
A.~Ashtekar, ``New hamiltonian formulation of general relativity,'' {\em Phys.
  Rev. D}, vol.~36, pp.~1587--1602, Sep 1987.

\bibitem{THIEMANN1996257}
T.~Thiemann, ``Anomaly-free formulation of non-perturbative, four-dimensional
  lorentzian quantum gravity,'' {\em Physics Letters B}, vol.~380, no.~3,
  pp.~257 -- 264, 1996.

\bibitem{rovelli1995discreteness}
C.~Rovelli and L.~Smolin, ``Discreteness of area and volume in quantum
  gravity,'' {\em Nuclear Physics B}, vol.~442, no.~3, pp.~593--619, 1995.

\bibitem{0264-9381-21-15-R01}
A.~Ashtekar and J.~Lewandowski, ``Background independent quantum gravity: a
  status report,'' {\em Classical and Quantum Gravity}, vol.~21, no.~15,
  p.~R53, 2004.

\bibitem{EH-Lim03}
A.~P.~C. Lim, ``{Area Operator in Loop Quantum Gravity},'' {\em Annales Henri
  Poincar{\'e}}, vol.~18(11), pp.~3719--3735, Jul 2017.

\bibitem{EH-Lim06}
A.~P.~C. Lim, ``{Invariants in Quantum Geometry},'' {\em Reports on
  Mathematical Physics}, vol.~87(1), pp.~87--105, 2021.

\bibitem{JACOBSON1988295}
T.~Jacobson and L.~Smolin, ``Nonperturbative quantum geometries,'' {\em Nuclear
  Physics B}, vol.~299, no.~2, pp.~295 -- 345, 1988.

\bibitem{EH-Lim04}
A.~P.~C. Lim, ``{Path Integral Quantization of Volume},'' {\em Annales Henri
  Poincar{\'e}}, vol.~21, pp.~1311--1327, 2020.

\bibitem{EH-Lim05}
A.~P.~C. Lim, ``{Quantized Curvature in Loop Quantum Gravity},'' {\em Reports
  on Mathematical Physics}, vol.~82(3), pp.~355--372, 2018.

\bibitem{EH-Lim02}
A.~P.~C. Lim, ``{Einstein-Hilbert Path Integrals in $\mathbb{R}^4$},'' {\em
  ArXiv e-prints}, Apr. 2017.

\bibitem{Geroch:1970uw}
R.~P. Geroch, ``{The domain of dependence},'' {\em J. Math. Phys.}, vol.~11,
  pp.~437--439, 1970.

\bibitem{PhysRevD.73.124038}
A.~Ashtekar, T.~Pawlowski, and P.~Singh, ``Quantum nature of the big bang: An
  analytical and numerical investigation,'' {\em Phys. Rev. D}, vol.~73,
  p.~124038, Jun 2006.

\bibitem{doi:10.1142/S0217732302006692}
A.~CORICHI, M.~P. RYAN, and D.~SUDARSKY, ``Quantum geometry as a relational
  construct,'' {\em Modern Physics Letters A}, vol.~17, no.~09, pp.~555--567,
  2002.

\bibitem{streater}
R.~F.~S. und A~S~Wightman, {\em PCT, Spin Statistics, And All That}.
\newblock New York, Amsterdam: W A Benjamin Inc., 1964.

\bibitem{PhysRevD.77.024046}
A.~Ashtekar, A.~Corichi, and P.~Singh, ``Robustness of key features of loop
  quantum cosmology,'' {\em Phys. Rev. D}, vol.~77, p.~024046, Jan 2008.

\bibitem{PhysRevLett.96.141301}
A.~Ashtekar, T.~Pawlowski, and P.~Singh, ``Quantum nature of the big bang,''
  {\em Phys. Rev. Lett.}, vol.~96, p.~141301, Apr 2006.

\bibitem{feynman2002feynman}
R.~Feynman, F.~Morinigo, W.~Wagner, and B.~Hatfield, {\em Feynman Lectures on
  Gravitation}.
\newblock Frontiers in Physics Series, Avalon Publishing, 2002.

\bibitem{peskin1995introduction}
M.~Peskin and D.~Schroeder, {\em An Introduction to Quantum Field Theory}.
\newblock Advanced book classics, Avalon Publishing, 1995.

\bibitem{MODANESE1995697}
G.~Modanese, ``Potential energy in quantum gravity,'' {\em Nuclear Physics B},
  vol.~434, no.~3, pp.~697 -- 708, 1995.

\bibitem{Nair}
V.~P. Nair, {\em Quantum field theory. A modern perspective}.
\newblock New York: Springer, 2005.

\bibitem{Smolin:1996fz}
L.~Smolin, ``{The Classical limit and the form of the Hamiltonian constraint in
  nonperturbative quantum general relativity},'' 1996.

\bibitem{PhysRevD.9.3292}
J.~D. Bekenstein, ``Generalized second law of thermodynamics in black-hole
  physics,'' {\em Phys. Rev. D}, vol.~9, pp.~3292--3300, Jun 1974.

\bibitem{PhysRevD.7.2333}
J.~D. Bekenstein, ``Black holes and entropy,'' {\em Phys. Rev. D}, vol.~7,
  pp.~2333--2346, Apr 1973.

\bibitem{wald1984general}
R.~Wald, {\em General Relativity}.
\newblock University of Chicago Press, 1984.

\bibitem{krasnov1996statistical}
V.~Krasnov, ``On statistical mechanics of gravitational systems,'' {\em Gen.
  Rel. Grav.}, vol.~30, no.~gr-qc/9605047, pp.~53--68, 1996.

\bibitem{PhysRevLett.77.3288}
C.~Rovelli, ``Black hole entropy from loop quantum gravity,'' {\em Phys. Rev.
  Lett.}, vol.~77, pp.~3288--3291, Oct 1996.

\bibitem{Krasnov:1996wc}
K.~V. Krasnov, ``{On Quantum statistical mechanics of Schwarzschild black
  hole},'' {\em Gen. Rel. Grav.}, vol.~30, pp.~53--68, 1998.

\bibitem{Ashtekar:1997yu}
A.~Ashtekar, J.~Baez, A.~Corichi, and K.~Krasnov, ``{Quantum geometry and black
  hole entropy},'' {\em Phys. Rev. Lett.}, vol.~80, pp.~904--907, 1998.

\bibitem{schutz1985first}
B.~Schutz, {\em A First Course in General Relativity}.
\newblock Series in physics, Cambridge University Press, 1985.

\bibitem{CS-Lim01}
A.~P.~C. Lim, ``Chern-{S}imons path integral on {$\Bbb R^3$} using abstract
  {W}iener measure,'' {\em Commun. Math. Anal.}, vol.~11, no.~2, pp.~1--22,
  2011.

\bibitem{Okon1}
E.~{Ekon}, ``{On the Status of the Equivalence Principle in Quantum Gravity},''
  {\em Preprint}, 2009.

\end{thebibliography}
\end{document}